\documentclass[lettersize,journal]{IEEEtran}
\usepackage{amsmath,amsfonts}
\usepackage{algorithmic}
\usepackage{array}
\usepackage[caption=false,font=normalsize,labelfont=sf,textfont=sf]{subfig}
\usepackage{textcomp}
\usepackage{stfloats}
\usepackage{url}
\usepackage{verbatim}
\usepackage{graphicx}
\def\BibTeX{{\rm B\kern-.05em{\sc i\kern-.025em b}\kern-.08em
    T\kern-.1667em\lower.7ex\hbox{E}\kern-.125emX}}
\usepackage{balance}
\usepackage{bm}
\usepackage{amsmath,amsthm}
\usepackage{tikz}
\usepackage{xcolor,graphicx}
\definecolor{myblue}{RGB}{183,221,232} 
\definecolor{myorange}{RGB}{251,215,187} 
\definecolor{mygreen}{RGB}{27,119,120} 
\usepackage{float}
\usepackage[normalem]{ulem}
\usepackage{booktabs}

\usepackage{subfig}
\definecolor{blue_matlab}{RGB}{0,48,189}

\newtheorem{lemma}{Lemma}[section]
\newtheorem{corollary}{Corollary}[section]
\newtheorem{Proposition}{Proposition}[section]
\newtheorem{Remark}{Remark}[section]

\begin{document}
\title{Scalable Small-Signal Stability Assessment of Power Systems Based on Frequency-Domain Quadratic Constraints}
\author{Xiaoying Liu, Linbin Huang, Hangyu Chen, and Huanhai Xin\vspace{-3mm}
% \thanks{Manuscript created October, 2020; This work was developed by the IEEE Publication Technology Department. This work is distributed under the \LaTeX \ Project Public License (LPPL) ( http://www.latex-project.org/ ) version 1.3. A copy of the LPPL, version 1.3, is included in the base \LaTeX \ documentation of all distributions of \LaTeX \ released 2003/12/01 or later. The opinions expressed here are entirely that of the author. No warranty is expressed or implied. User assumes all risk.}
\thanks{This work was supported by the Smart Grid National Science and Technology Major Project of China under Grant 2026ZD0812903. 

The authors are with the College of Electrical Engineering, Zhejiang University, Hangzhou 310027, China. (E-mail: hlinbin@zju.edu.cn).}
}

\markboth{Journal of \LaTeX\ Class Files}%
{How to Use the IEEEtran \LaTeX \ Templates}

\maketitle

\begin{abstract}
Small-signal stability analysis of power electronics-dominated power systems is becoming increasingly challenging due to the large-scale integration of heterogeneous grid-following (GFL) and grid-forming (GFM) converters. Classical centralized methods, such as eigenvalue analysis and the generalized Nyquist criterion, provide direct stability assessment tools but exhibit limited scalability as the system size increases. This motivates scalable analysis methods for stability assessment and device-level diagnosis, and calls for less conservative stability conditions. To this end, this paper proposes a scalable small-signal stability assessment method with quantitative stability indices for multi-converter systems based on frequency-domain quadratic constraints (FQCs). The FQCs can be thought of as a special case of the integral quadratic constraint (IQC) for linear systems. It is shown that several existing stability conditions, including geometric conditions based on the Davis-Wielandt (DW) shell, numerical range, $x$-$z$ graph, scaled relative graph (SRG), mixed gain-phase condition, and passivity condition, are special cases of the FQC-based stability condition. This FQC-based formulation also clarifies the requirements for decentralized verification of these existing stability conditions. To further reduce conservatism, a full-multiplier FQC-based stability condition is developed, based on which an optimization problem is formulated for scalable stability certification and device-level diagnosis using quantitative stability indices without relying on graphical inspection. Combined with the mixed gain-phase condition, this optimization problem forms a hierarchical screening-diagnosis procedure that improves efficiency. Case studies demonstrate that the proposed method is less conservative than existing criteria and can effectively identify problematic devices associated with potential instability risks.
\end{abstract}

\vspace{-2mm}

\section{Introduction}

\IEEEPARstart{M}{ODERN} power systems are undergoing a profound transformation driven by the large-scale integration of power electronics-interfaced resources, such as renewable generation, energy storage systems, and voltage-source converter-based HVDC links~\cite{ref_context_1}. As a result, system dynamics are no longer dominated solely by conventional synchronous generators (SGs), but are increasingly shaped by heterogeneous grid-following (GFL) and grid-forming (GFM) converters, whose dynamic characteristics differ substantially from those of SGs~\cite{ref_context_2}. The coexistence of these heterogeneous dynamics and their coupling through the power network can introduce new small-signal stability issues, complicating stability assessment in large-scale power systems.

Classical small-signal stability assessment methods, such as eigenvalue analysis~\cite{ref_eigen_CCM} and the generalized Nyquist criterion~\cite{ref_GNC}, provide direct and accurate tools for analyzing the stability of interconnected power systems. However, eigenvalue analysis requires a detailed state-space model of the entire closed-loop system. As the number of heterogeneous devices increases, such a centralized model becomes high-dimensional and cumbersome to construct, update, and analyze~\cite{ref_eigen_drawback}. In comparison, the generalized Nyquist criterion enables device or subsystem dynamics to be represented by frequency-domain impedance models, which can be obtained through frequency scanning without detailed internal information~\cite{ref_GNC_measure}. Nevertheless, it requires the evaluation of intricate characteristic loci, thereby limiting its scalability~\cite{ref_GNC_drawback}.
%Moreover, these centralized approaches often provide limited insight into which individual devices contribute to potential instability, making them less suitable for scalable stability assessment and device-level diagnosis in large-scale power electronics-dominated power systems.

Decentralized stability analysis methods address the scalability limitations of centralized approaches by relying on local device or subsystem dynamics rather than constructing and analyzing the entire closed-loop system~\cite{ref_decentralized_condition}. Representative examples include passivity-based analysis~\cite{ref_passivity}, the small-gain theorem~\cite{ref_Gain_condition}, and mixed gain-phase conditions~\cite{ref_Gain_Phase_1,ref_Gain_Phase_2}, which have been applied to study power electronics-dominated systems. Despite their scalability, these conditions may remain considerably conservative and may fail to certify stability in certain frequency ranges. 

%To reduce this conservatism, geometric stability conditions have recently attracted increasing attention. Among these approaches, the Davis-Wielandt (DW) shell provides a three-dimensional geometric representation of MIMO frequency-response matrices that retains richer gain- and phase-related information than lower-dimensional descriptions~\cite{ref_DW_shell}, such as the numerical range~\cite{ref_phantom_DW_shell}, the $x$-$z$ graph~\cite{ref_x_z_graph}, the scaled relative graph (SRG)~\cite{ref_SRG_definition,ref_SRG_decentralized}, gain-phase conditions, and passivity-based analysis. As shown in~\cite{ref_phantom_DW_shell}, these lower-dimensional descriptions are projections of the DW shell, and the DW shell condition is the least conservative compared with them. 

To reduce this conservatism, geometric stability conditions based on tools such as the Davis-Wielandt (DW) shell~\cite{ref_DW_shell}, the numerical range~\cite{ref_phantom_DW_shell}, the \(x\)-\(z\) graph~\cite{ref_x_z_graph}, and the scaled relative graph (SRG)~\cite{ref_SRG_definition,ref_SRG_decentralized} have recently attracted increasing attention. Among these tools, the DW shell provides a three-dimensional geometric representation of MIMO frequency-response matrices, whereas the other descriptions can be interpreted as lower-dimensional projections of the DW shell~\cite{ref_phantom_DW_shell}. Hence, the DW shell-based condition is generally the least conservative one among these geometric conditions. These geometric conditions have been developed for both centralized and decentralized stability assessment, which require geometric representations of both the device-side and network-side open-loop systems for graphical comparison and verification. For instance, recent studies have proposed decentralized SRG-based and $x$-$z$ graph-based stability certificates~\cite{ref_x_z_graph,ref_SRG_decentralized} for power systems by means of graph separation. 
However, graphical conditions might be difficult to check when a large number of devices are involved, which motivates scalable and quantitative stability conditions with minimum conservatism.

Recent studies have attempted to integrate various stability conditions within a unified stability assessment framework for power electronics-dominated power systems. For instance, the stability condition in~\cite{ref_unified_framework} can recover several existing criteria, including the small gain theorem, the passivity theorem, and the mixed small gain-phase conditions, which enables the construction of less conservative stability indicators. 
%However, this framework does not explicitly incorporate geometric stability conditions based on the DW shell and its lower-dimensional representations. 
Recently, Refs.~\cite{ref_FQC_2} and~\cite{ref_Hallinan_PartI} developed analysis methods that yield decentralized stability conditions for DC and AC grids based on the integral quadratic constraint (IQC) theory. 
% These developments motivate a unified and scalable FQC framework for seeking less conservative stability certificates and enabling quantitative device-level diagnosis in large-scale multi-converter systems.
However, it remains unclear whether there is a framework to generalize over graphical conditions (for instance,  based on DW shell, numerical range, and $x$-$z$ graph) and non-graphical conditions such as those in~\cite{ref_unified_framework} and~\cite{ref_Hallinan_PartI}. More importantly, the problem of device-level diagnosis remains unsolved when non-graphical conditions are used in multi-converter systems.

This paper proposes a frequency-domain quadratic constraint (FQC) based method for scalable small-signal stability assessment and device-level diagnosis of multi-converter systems. As will be shown below, the FQC is a special case of the well-known IQC, customized for linear systems, while it recovers many graphical and non-graphical stability conditions. The main contributions are summarized as follows.

\begin{enumerate}{}{}
\item{An FQC framework is developed to unify existing stability conditions, under which the requirements for decentralized verification and the relative conservatism of different conditions are systematically revealed.}
\item{A low-conservatism FQC-based stability condition is developed, enabling an optimization problem for stability certification and device-level diagnosis using quantitative stability indices without graphical inspection.}
\item{A hierarchical screening-diagnosis procedure is developed by integrating this optimization problem with the mixed gain-phase condition to improve the efficiency of scalable small-signal stability assessment.}
\end{enumerate}

The rest of the paper is organized as follows: Section~\ref{sec:Modeling} presents the modeling of multi-converter systems. Section~\ref{sec:existing_decentralized_stability_conditions} relates existing stability conditions to the FQC framework and discusses the requirements for their decentralized verification. Section~\ref{sec:screening-diagnosis} develops the optimization-based scalable stability assessment procedure. Detailed case analysis is provided in Section~\ref{sec:case}. Section~\ref{sec:conclusion} concludes the paper.

\it{Notation}\rm: Let $\mathbb{C}$, $\mathbb{C}^n$, and $\mathbb{C}^{p\times m}$ denote the sets of complex numbers, $n$-dimensional complex vectors, and $p\times m$ complex matrices, respectively, with the corresponding real sets denoted analogously by replacing $\mathbb{C}$ with $\mathbb{R}$. For a matrix $\bm{M}$, we use $\overline{\bm{M}}$ and $\bm{M}^*$ to denote its conjugate and conjugate transpose, respectively, and $\det(\bm{M})$ to denote its determinant. A complex matrix $\bm{M}$ can be decomposed as $\bm{M}=\Re(\bm{M})+j\Im(\bm{M})$, where $\Re(\bm{M}):=\frac{1}{2}(\bm{M}+\bm{M}^*)$ denotes the Hermitian part and $\Im(\bm{M}):=\frac{1}{2j}(\bm{M}-\bm{M}^*)$ denotes the skew-Hermitian part. The inner product of $\bm{u},\bm{y}\in\mathbb{C}^n$ is defined by $\langle\bm{u},\bm{y}\rangle:=\bm{u}^*\bm{y}$, and the Euclidean norm for $\bm{u}$ is defined by $\|\bm{u}\|=\sqrt{\langle\bm{u},\bm{u}\rangle}$. We use the compact notation $(z,\nu)\in\mathbb{C}\times\mathbb{R}_+$ to denote the point $(\Re(z),\Im(z),\nu)$ in a three-dimensional real coordinate space. We use $\bm{I}$ to denote an identity matrix of compatible dimension. We use $\otimes$ to denote the Kronecker product. We denote a diagonal (or block-diagonal) matrix $\bm{D}$ as $\text{diag}(D_1,\ldots, D_n)$ where $D_1,\ldots, D_n$ are the diagonal elements (or blocks).

\section{Modeling of Multi-Converter Power Systems}\label{sec:Modeling}

Fig.~\ref{fig:power_system} illustrates a typical multi-converter power system, where heterogeneous converters are interconnected through the power network. The system can be naturally decomposed into the device and network sides, where the former includes converters with different control structures and parameters (e.g., GFL or GFM), and the latter consists of transmission lines and conventional loads.

\begin{figure}[t]
\centering
\includegraphics[width=0.8\columnwidth]{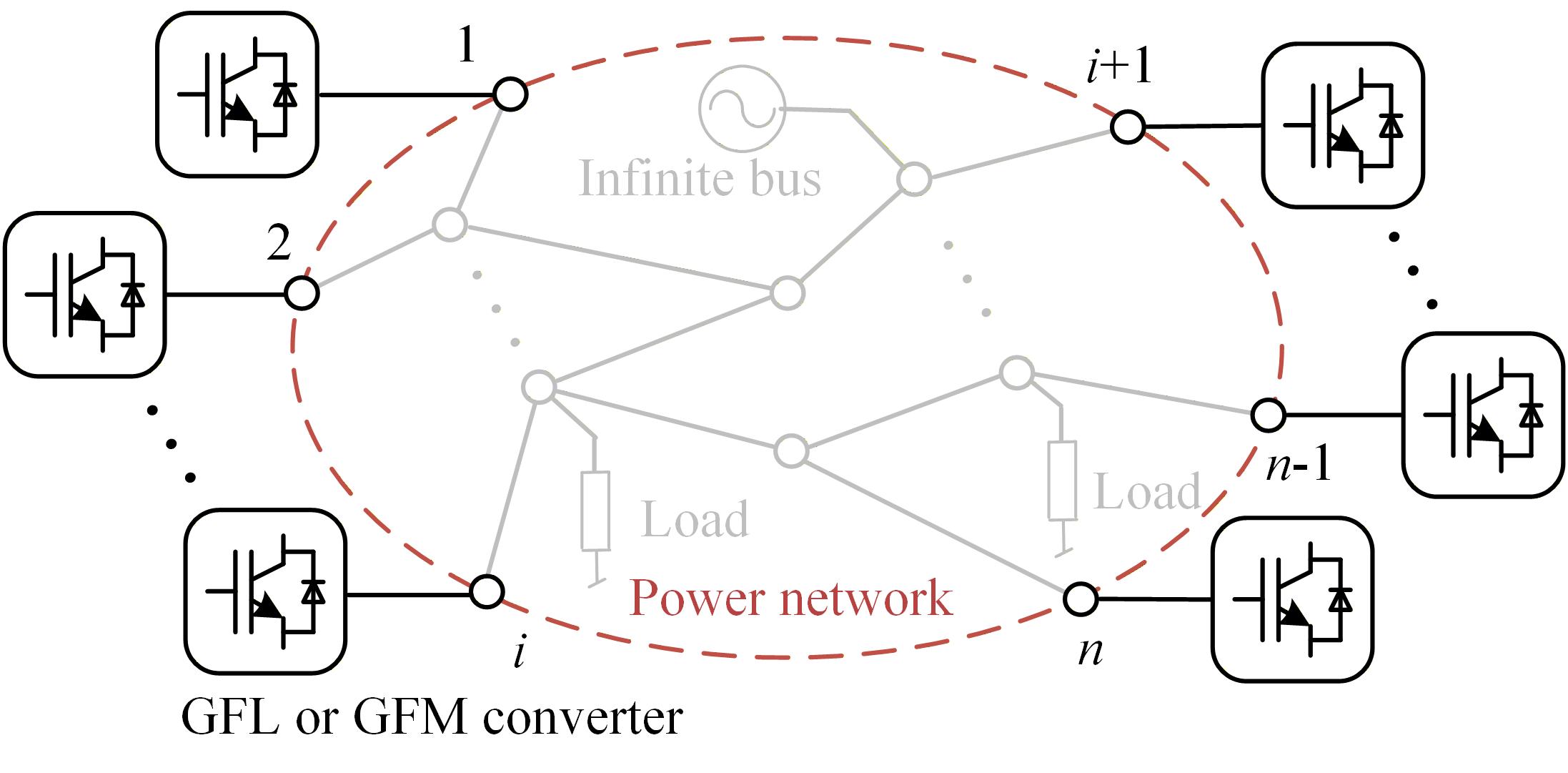}
\vspace{-4mm}
\caption{Illustration of a multi-converter power system.}
\label{fig:power_system}
\end{figure}

The small-signal dynamics of both sides are represented by frequency-domain admittance models, which describe the voltage-current perturbation relationship in the \(dq\) coordinate. In particular, the admittance model of all converters can be expressed by a block-diagonal matrix
\begin{equation}\label{Equ_converter_model}
\bm{Y}_c(s)=\text{diag}(\bm{Y}_{c,1}(s),\cdots,\bm{Y}_{c,n}(s)),
\end{equation}
where $\bm{Y}_{c,i}(s)$ denotes the $2\times 2$ admittance matrix of the $i$-th converter in a global $dq$ coordinate ($i=1,\cdots,n$).
The admittance matrix $\bm{Y}_g(s)$ and the impedance matrix $\bm{Z}_g(s) = \bm{Y}^{-1}_g(s)$ in the global $dq$ coordinate represent the small-signal dynamics of the power network, including the resistors, inductors, and shunt capacitors of transmission lines, as well as conventional loads. Note that the power network matrix $\bm{Y}_g(s)$ is obtained by Kron reduction of the original power network while retaining the converter buses. In the small-signal model, infinite buses can be treated as grounded because their voltage perturbations are zero. The detailed derivations of $\bm{Y}_{c,i}(s)$ and $\bm{Y}_g(s)$ can be found in~\cite{ref_Gain_Phase_2,ref_modeling}.

Thus, the multi-converter system is formulated as the negative feedback interconnection of the open-loop systems $\bm{Y}_c(s)$ and $\bm{Z}_g(s)$, as shown in Fig.~\ref{fig:interconnection}. Conventional centralized methods assess stability by checking whether the characteristic equation $\det(\bm{I}+\bm{Y}_c(s)\bm{Z}_g(s))=0$ has right-half-plane solutions. In contrast, scalable stability conditions certify stability through local converter-level conditions and a compatible network-side condition, as illustrated in the next section.

\begin{figure}
\centering
\includegraphics[width=\columnwidth]{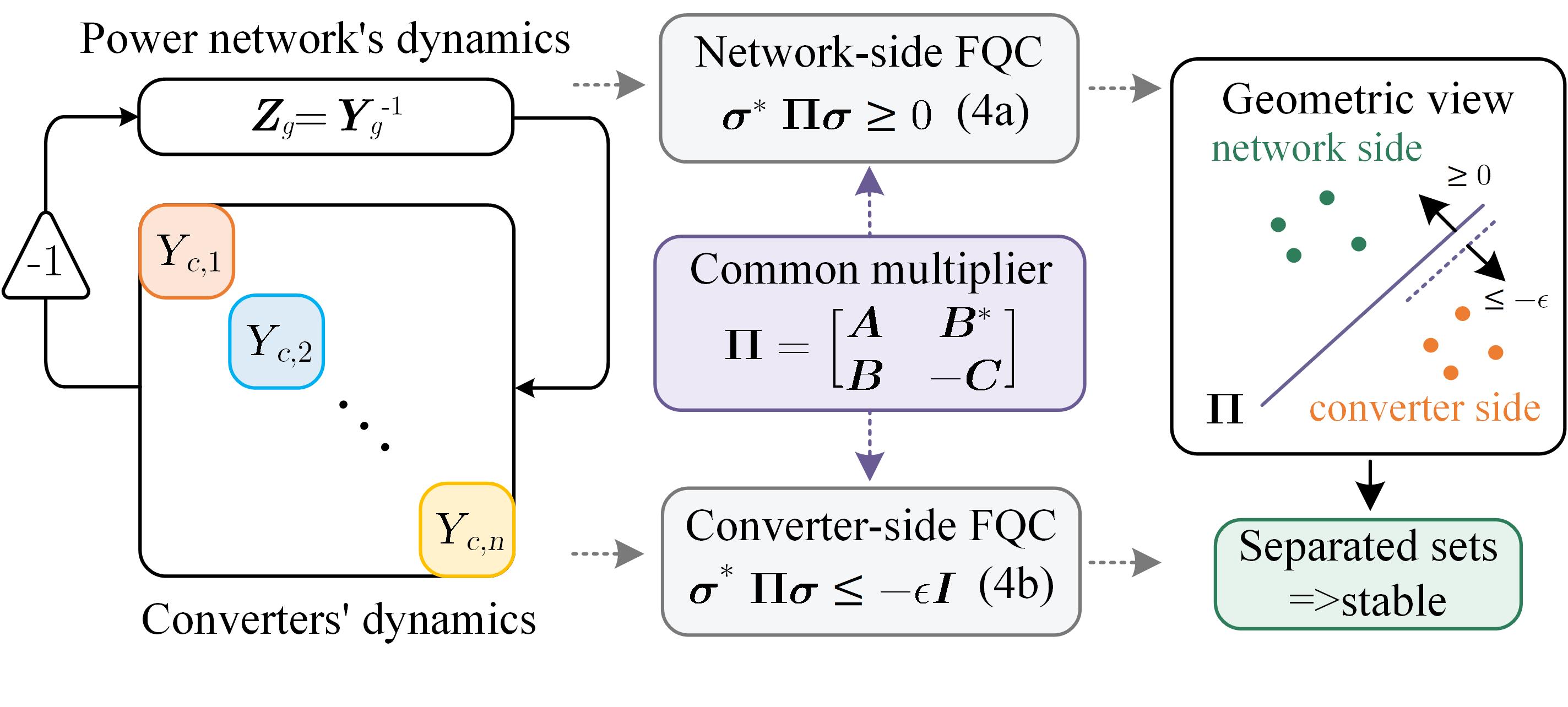}
\vspace{-10mm}
\caption{FQC-based stability condition for the closed-loop system.}
\label{fig:interconnection}
\vspace{-2mm}
\end{figure}

\section{Stability Conditions Covered by FQCs}\label{sec:existing_decentralized_stability_conditions}

This section first introduces the FQC-based stability condition. Then, complex matrix analysis tools associated with existing stability analysis methods are introduced to support the subsequent comparative discussion. Finally, we show how various existing stability conditions can be recovered as special cases of the FQC framework and discuss the requirements for their decentralized verification.

\subsection{FQC-based Stability Condition}

An integral quadratic constraint (IQC) on an input--output pair $(\bm w,\bm z)$ can be expressed in the frequency domain as~\cite{ref_FQC_1}
\begin{equation}\label{Equ_IQC_background}
\int_{-\infty}^{\infty}\begin{bmatrix}\hat{\bm w}(j\omega)\\\hat{\bm z}(j\omega)\end{bmatrix}^{*}\bm\Pi(j\omega)\begin{bmatrix}\hat{\bm w}(j\omega)\\\hat{\bm z}(j\omega)\end{bmatrix}\,d\omega \geq 0,
\end{equation}
where \(\hat{\bm w}\) and \(\hat{\bm z}\) denote the Fourier transforms of \(\bm w\) and \(\bm z\), respectively, and \(\bm\Pi(j\omega)\) is Hermitian. For an LTI subsystem, $\hat{\bm z}(j\omega)=\bm G(j\omega)\hat{\bm w}(j\omega)$. Substituting this relation into~\eqref{Equ_IQC_background} gives $\int_{-\infty}^{\infty}\hat{\bm w}^{*}(j\omega)\bm M(j\omega)\hat{\bm w}(j\omega)\,d\omega\geq0$, where \(\bm M(j\omega)=\bigl[\begin{smallmatrix}\bm I\\\bm G(j\omega)\end{smallmatrix}\bigr]^{*}\bm\Pi(j\omega)\bigl[\begin{smallmatrix}\bm I\\\bm G(j\omega)\end{smallmatrix}\bigr]\). If $\bm M(j\omega)\succeq0$ at every frequency, the integral inequality follows. We refer to such pointwise quadratic matrix inequalities as FQCs. In the following lemma, \(\bm G(j\omega)=\tau\bm Z_g(j\omega)\) is used to formulate the network-side FQC, while a strict FQC is imposed on the converter model using the same multiplier \(\bm\Pi(j\omega)\). This leads to a stability condition for the interconnected network and converter models, following the result established in~\cite{ref_FQC_2}.

\begin{lemma}[FQC-based stability condition]  \label{lem:QC_stability}
Suppose that the open-loop systems $\bm{Y}_c(s)$ and $\bm{Z}_g(s)$ are stable. Then, the negative feedback interconnection of $\bm{Y}_c(s)$ and $\bm{Z}_g(s)$ is stable if for all $\omega\in[0,\infty)$ and $\tau\in(0,1]$, there exist a Hermitian matrix $\bm{\Pi}(j\omega,\tau)$ and a scalar $\epsilon(\omega,\tau)>0$ such that 
\begin{subequations}
\begin{align}
\left[\begin{array}{c}\bm{I}\\ \tau \bm{Z}_g(j\omega)\end{array}\right]^{*}\,\bm{\Pi}(j\omega,\tau)\,\left[\begin{array}{c}\bm{I}\\ \tau \bm{Z}_g(j\omega)\end{array}\right] &\succeq 0,
\label{IQC_1}\\
\left[\begin{array}{c}-\bm{Y}_c(j\omega) \\ \bm{I} \end{array}\right]^{*}\,\bm{\Pi}(j\omega,\tau)\,\left[\begin{array}{c}-\bm{Y}_c(j\omega) \\ \bm{I} \end{array}\right] &\preceq -\epsilon(\omega,\tau)\bm{I}.
\label{IQC_2}
\end{align}
\end{subequations}
\end{lemma}

\begin{proof}
The proof is given in Appendix~\ref{app:proof_1}.
\end{proof}

For subsequent analysis, it is convenient to parameterize the FQC multiplier \(\bm{\Pi}(j\omega,\tau)\) in a block-partitioned form \(\bm{\Pi}(j\omega,\tau)=\bigl[\begin{smallmatrix}\bm{A}(\omega,\tau) & \bm{B}(\omega,\tau)^* \\ \bm{B}(\omega,\tau) & -\bm{C}(\omega,\tau)\end{smallmatrix}\bigr]\), which facilitates the expansion of~\eqref{IQC_1} and~\eqref{IQC_2} and leads to the following corollary.

\begin{corollary}[ ]  \label{coro:QC_stability}
Suppose that the open-loop systems $\bm{Y}_c(s)$ and $\bm{Z}_g(s)$ are stable. Then, the negative feedback interconnection of $\bm{Y}_c(s)$ and $\bm{Z}_g(s)$ is stable if for all $\omega\in[0,\infty)$ and $\tau\in(0,1]$, there exist $\bm{A}(\omega,\tau)=\bm{A}(\omega,\tau)^*$, $\bm{B}(\omega,\tau)$, $\bm{C}(\omega,\tau)=\bm{C}(\omega,\tau)^*$ and $\epsilon(\omega,\tau)>0$ such that 
\begin{subequations}
\begin{align}
\begin{aligned}
& \bm{A}(\omega,\tau)+\tau\bm{Z}_g(j\omega)^*\bm{B}(\omega,\tau)+\tau\bm{B}(\omega,\tau)^*\bm{Z}_g(j\omega)\\
& -\tau^2\bm{Z}_g(j\omega)^*\bm{C}(\omega,\tau)\bm{Z}_g(j\omega) \succeq 0,
\end{aligned}
\label{IQC_3}\\
\begin{aligned}
& \bm{Y}_c(j\omega)^*\bm{A}(\omega,\tau)\bm{Y}_c(j\omega)-\bm{B}(\omega,\tau)\bm{Y}_c(j\omega)\\
& -\bm{Y}_c(j\omega)^*\bm{B}(\omega,\tau)^*-\bm{C}(\omega,\tau)\preceq -\epsilon(\omega,\tau)\bm{I}.
\end{aligned}
\label{IQC_4}
\end{align}
\end{subequations}
\end{corollary}

The open-loop stability of $\bm{Y}_c(s)$ and $\bm{Z}_g(s)$ is generally guaranteed, as each converter is designed to remain stable when connected to an infinite bus and the power network is typically passive. According to Corollary~\ref{coro:QC_stability}, the closed-loop stability can be certified by finding a matrix triplet $\bm{A}(\omega,\tau)$, $\bm{B}(\omega,\tau)$, and $\bm{C}(\omega,\tau)$ satisfying~\eqref{IQC_3} and~\eqref{IQC_4}, as shown in Fig.~\ref{fig:interconnection}. The freedom in choosing this triplet enhances the flexibility of this stability condition and offers potential for scalable stability analysis.
The FQC-based condition also admits a geometric interpretation. For each fixed $(\omega,\tau)$, the same matrix triplet is used in the network-side and converter-side inequalities, yielding a separation between the network-side set and the converter-side set in geometric representations such as the DW shell, numerical range, $x$-$z$ graph, and SRG. These graphical forms are analyzed in the following subsections.

\subsection{Geometric Descriptions of Complex Matrices}
To support the analysis of existing geometric separation conditions in the next subsection, we first introduce several complex matrix analysis tools that are used to characterize the converter-side and network-side sets.

Given a matrix $\bm{M}\in\mathbb{C}^{n\times n}$, its Davis-Wielandt (DW) shell, numerical range, and $x$-$z$ graph are defined as~\cite{ref_x_z_graph,ref_DW_shell_definition}
\begin{equation}
\label{Equ_DW}
\begin{aligned}
&\mathcal{DW}(\bm{M})=\{(\bm{u}^*\bm{M}\bm{u},\|\bm{M}\bm{u}\|^2):\bm{u}\in\mathbb{C}^n,\ \|\bm{u}\|=1\},\\
&\mathcal{W}(\bm{M})=\{\bm{u}^*\bm{M}\bm{u}:\bm{u}\in\mathbb{C}^n,\|\bm{u}\|=1\},\\
&\mathcal{P}(\bm{M})=\{(\Re(\bm{u}^*\bm{M}\bm{u}),\|\bm{M}\bm{u}\|^2):\bm{u}\in\mathbb{C}^n,\|\bm{u}\|=1\},
\end{aligned}
\end{equation}
where $\bm{u}^*\bm{M}\bm{u}\in\mathbb{C}$ and $\|\bm{M}\bm{u}\|^2\in\mathbb{R}_+$.

The scaled relative graph (SRG) of $\bm{M}$ is defined as
\begin{equation}
\label{Equ_SRG}
\mathrm{SRG}(\bm{M})=\left\{
\frac{\lVert \bm{y}\rVert}{\lVert \bm{u}\rVert}
\exp\!\left(
\pm j\,\arccos\!\left(
\frac{\Re(\langle \bm{y},\bm{u}\rangle)}{\lVert \bm{y}\rVert\,\lVert \bm{u}\rVert}
\right)
\right)
\right\},
\end{equation}
where $\|\bm{u}\|=1$, and $\bm{y}=\bm{M}\bm{u}$. Moreover, the SRG of $\bm{M}$ can be calculated by $g(\mathcal{W}\!\bigl(f(\bm{M})\bigr))$~\cite{ref_SRG_definition}, where
\begin{equation}
\label{com_SRG}
\begin{aligned}
& f(\bm{M})=\bm{T}^{-1}\bm{K}\bm{T}^{-1},\bm{T} = \left(\bm{I}+\bm{M}^{*}\bm{M}\right)^{\frac12},\\
& \bm{K}=\bm{M}^{*}\bm{M}-j(\bm{M}^{*}+\bm{M})-\bm{I},\\
& z\in \mathcal{W}\!\bigl(f(\bm{M})\bigr),g(z)=\left\{
\frac{\Im(z)\pm j\sqrt{1-|z|^{2}}}{\Re(z)-1}
\right\}.
\end{aligned}
\end{equation}

The gains of the complex matrix $\bm{M}$ can be defined by its singular values $\sigma(\bm{M}):=[\sigma_1(\bm{M})\ \cdots\ \sigma_n(\bm{M})]$, where $\sigma_1(\bm{M})\geq\cdots\geq\sigma_n(\bm{M})$. The phase of the complex matrix $\bm{M}$ can be defined from its numerical range. A matrix is said to be sectorial if $0\notin\mathcal{W}(\bm{M})$. In this case, there exist a nonsingular matrix $\bm{T}$ and a diagonal unitary matrix $\bm{D}$ such that $\bm{M}=\bm{T}^*\bm{D}\bm{T}$~\cite{ref_phase_definition}. Then, the phases of the sectorial matrix $\bm{M}$ are defined as the phases of the $n$ diagonal entries of the matrix $\bm{D}$, ordered as \(\phi_1(\bm{M})\geq\cdots\geq\phi_n(\bm{M})\).

\subsection{Comparison with Existing Stability Conditions}\label{sec:Comparison with Existing Stability Conditions}
This subsection relates the stability conditions in~\cite{ref_phantom_DW_shell,ref_x_z_graph,ref_SRG_definition,ref_mixed_Gain_Phase} to the FQC framework. By setting $\bm{A}(\omega,\tau)=a(\omega,\tau)\bm{I}$, $\bm{B}(\omega,\tau)=b(\omega,\tau)\bm{I}$, $\bm{C}(\omega,\tau)=c(\omega,\tau)\bm{I}$ in Corollary~\ref{coro:QC_stability}, we obtain
\begin{subequations}\label{eq_IQC_all}
\begin{align}
\begin{aligned}
& a(\omega,\tau)\bm{I}+\tau b(\omega,\tau)\bm{Z}_g(j\omega)^*+\tau b(\omega,\tau)^*\bm{Z}_g(j\omega)\\
& -\tau^2c(\omega,\tau)\bm{Z}_g(j\omega)^*\bm{Z}_g(j\omega) \succeq 0,
\end{aligned}
\label{IQC_DW_Zg}\\
\begin{aligned}
& a(\omega,\tau)\bm{Y}_c(j\omega)^*\bm{Y}_c(j\omega)-b(\omega,\tau)\bm{Y}_c(j\omega)\\
& -b(\omega,\tau)^*\bm{Y}_c(j\omega)^*-c(\omega,\tau)\bm{I}\preceq -\epsilon(\omega,\tau)\bm{I}.
\end{aligned}
\label{IQC_DW_Yc}
\end{align}
\end{subequations}

For notational simplicity, we omit the dependence of the multipliers $a$, $b$, $c$, and auxiliary quantities on the fixed pair $(\omega,\tau)$ in the remainder of this paper unless otherwise stated. Different restrictions on the scalar multipliers \(a\), \(b\), and \(c\) recover existing stability conditions through their associated geometric separation tests.

A feasible triplet \(a,c\in\mathbb R\), \(b\in\mathbb C\) satisfying \eqref{IQC_DW_Zg} and \eqref{IQC_DW_Yc} defines a separating hyperplane between the converter and inverse network DW shells
\[H_d=\{(z,\nu)\in\mathbb{C}\times\mathbb{R}_+:2\Re(bz)+a\nu=x\},\]
where \(z\) and \(\nu\) denote the complex numerical-range and squared-gain coordinates of the DW shell representation in \eqref{Equ_DW}, respectively. The separating threshold \(x\) is chosen such that \(c-\epsilon<x<c\).

Setting \(a=0\) gives numerical range separation, while restricting \(b\) to be real gives the \(x\)-\(z\) graph and SRG separation conditions. The choices \(a>0,\ b=0\) and \(a=c=0\) yield the small-gain and small-phase conditions, respectively. Table~\ref{tab:comparison} summarizes these multiplier structures and their separating boundaries, with detailed derivations provided in Appendix~\ref{app:exist_stability_conditions}. The following proposition collects the corresponding sufficient conditions for closed-loop stability.

\begin{table*}[t]
\centering
\caption{Comparison of Stability Conditions Derived from the Scalar FQCs}
\label{tab:comparison}
%\begin{tabular}{p{0.22\linewidth}ccc}
\begin{tabular}{cccc}
\toprule
\multicolumn{4}{c}{
Scalar FQC-based stability condition
}\\
\midrule
\multicolumn{4}{c}{
$ a\bm{I}+\tau b\bm{Z}_g(j\omega)^*+\tau b^*\bm{Z}_g(j\omega)-\tau^2c\bm{Z}_g(j\omega)^*\bm{Z}_g(j\omega)\succeq 0, \qquad a\bm{Y}_c(j\omega)^*\bm{Y}_c(j\omega)-b\bm{Y}_c(j\omega)-b^*\bm{Y}_c(j\omega)^*-c\bm{I}\preceq-\epsilon\bm{I}.$
}\\
\midrule
Stability Condition & Scalar Multiplier & Separating Boundary & Conservatism \\
\midrule
DW shell separation & $a\in\mathbb{R}$, $b\in\mathbb{C}$, $c\in\mathbb{R}$ & $H_d=\{(z,\nu)\in\mathbb{C}\times\mathbb{R}_+:2\Re(bz)+a\nu=x,a\in\mathbb{R},b\in\mathbb{C}\}$ & Lowest \\
Numerical range separation & $a=0$, $b\in\mathbb{C}$, $c\in\mathbb{R}$ & $H_n=\{z\in\mathbb{C}:2\Re(bz)=x,b\in\mathbb{C}\}$ & Low \\
$x$-$z$ graph separation & $a\in\mathbb{R}$, $b\in\mathbb{R}$, $c\in\mathbb{R}$ & $H_p=\{(\Re(z),\nu)\in\mathbb{R}\times \mathbb{R}_+:2b\Re(z)+a\nu=x,a,b\in\mathbb{R}\}$ & Low \\
SRG separation & $a\in\mathbb{R}$, $b\in\mathbb{R}$, $c\in\mathbb{R}$ & $H_{r}=\{r\in\mathbb{C}:2b\Re(r)+a|r|^2=x,\ a,b\in\mathbb{R}\}$ & Low \\
Small gain theorem & $a>0$, $b=0$, $c\in\mathbb{R}$ & $H_\sigma=\{v\in\mathbb{R}_+:av=x,a\in\mathbb{R}\}$ & High \\
Small phase theorem & $a=c=0$, $b=b_1e^{jb_2}\in\mathbb{C}$ & $H_{\phi}^{\pm}=\{z^\phi\in\mathbb{C}:\Re(z^\phi)\cos b_2\pm \Im(z^\phi)\sin b_2=0,\ b_2\in\mathbb R\}$ & Medium \\
\bottomrule
\end{tabular}
\par\vspace{1mm}
\begin{minipage}{0.95\textwidth}
\footnotesize \emph{Note:} The SRG separation condition and the $x$-$z$ graph separation condition have the same level of conservatism.
\end{minipage}
\vspace{-3mm}
\end{table*}

\begin{Proposition}[ ]  \label{pro:exist_stability_conditions}
Suppose that the open-loop systems $\bm{Y}_c(s)$ and $\bm{Z}_g(s)$ are stable. Then, the negative feedback interconnection of $\bm{Y}_c(s)$ and $\bm{Z}_g(s)$ is stable if for all $\omega\in[0,\infty)$ and $\tau\in(0,1]$, either

i) Conditions~\eqref{IQC_3} and~\eqref{IQC_4} hold, or

ii) $\mathcal{DW}(-\bm{Y}_c(j\omega))\cap\mathcal{DW}^{-1}(\tau\bm{Z}_g(j\omega))=\emptyset$, or

iii) $\mathcal{W}(-\bm{Y}_c(j\omega))\cap\mathcal{W}^{-1}(\tau\bm{Z}_g(j\omega))=\emptyset$, or

iv) $\mathcal{P}(-\bm{Y}_c(j\omega))\cap\mathcal{P}^{-1}(\tau\bm{Z}_g(j\omega))=\emptyset$, or

v) $\mathrm{SRG}(-\bm{Y}_c(j\omega))\cap\mathrm{SRG}^{-1}(\tau\bm{Z}_g(j\omega))=\emptyset$, or

vi) $\sigma_1(\bm{Y}_c(j\omega))\sigma_1(\bm{Z}_g(j\omega))<1$, or 

vii) $\bm{Y}_c(j\omega)$ and $\bm{Z}_g(j\omega)$ are sectorial, $\phi_1(\bm{Y}_c(j\omega))+\phi_1(\bm{Z}_g(j\omega))<\pi$ and $\phi_{2n}(\bm{Y}_c(j\omega))+\phi_{2n}(\bm{Z}_g(j\omega))>-\pi$.
\end{Proposition}

\begin{Remark}[ ]  \label{remark:SRG}
The \(x\)-\(z\) graph and SRG represent the same information in different coordinates. Specifically, for each point $(\Re(z),\nu)$ in the $x$-$z$ graph with \(z=\bm u^*\bm M\bm u\) and \(\nu=\|\bm M\bm u\|^2\), the corresponding SRG point $r$ satisfies $\Re(r)=\Re(z)$, $|r|^2=\nu$ and $r=\Re(z)\pm j\sqrt{\nu-\Re(z)^2}$ according to \eqref{Equ_SRG}. Conversely, each SRG point gives back $\Re(z)=\Re(r)$ and $\nu=|r|^2$. Therefore, the SRG separating boundary is a circle centered on the real axis, as shown in Table~\ref{tab:comparison} and Fig.~\ref{fig:existing_conditions}\subref{fig:existing_conditions_srg}. The conjugate SRG points represent the same \(x\)-\(z\) graph point. Despite their different boundary shapes, the two representations yield equivalent separation conditions with the same conservatism. Both correspond to \(b\in\mathbb R\), which discards \(\Im(z)\) and reduces the flexibility available in DW shell separation.
\end{Remark}

\begin{Remark}[ ]  \label{remark:gain}
As summarized in Table~\ref{tab:comparison}, the small gain theorem compares converter and network gain bounds without using the numerical-range coordinate. It is therefore a specialization of \(x\)-\(z\) graph separation that uses only gain information, as shown in Fig.~\ref{fig:existing_conditions}\subref{fig:existing_conditions_xz}.
\end{Remark}

\begin{Remark}[ ]  \label{remark:phase}
Table~\ref{tab:comparison} shows that the small phase theorem restricts the half-plane boundaries to lines passing through the origin. Numerical-range separation allows an affine separating line with variable orientation and offset, and therefore provides more flexibility, even when the matrices are not sectorial, as shown in Fig.~\ref{fig:existing_conditions}\subref{fig:existing_conditions_nr}. 
\end{Remark}

Based on the above analysis, the FQC framework provides a unified perspective to interpret existing stability conditions through different restrictions on the scalar multipliers $a$, $b$, and $c$. Among the conditions in Proposition~\ref{pro:exist_stability_conditions}, the DW-shell separation condition adopts the most flexible scalar-multiplier structure ($a\in\mathbb{R}$, $b\in\mathbb{C}$, $c\in\mathbb{R}$). By imposing additional restrictions on these multipliers, it naturally degenerates into more conservative conditions, as summarized in Table~\ref{tab:comparison}. Therefore, these conditions form a hierarchy in which increasing structural simplicity is achieved at the cost of reduced flexibility and increased conservatism.

\begin{figure*}[t]
    \centering
    \subfloat[]{%
        \includegraphics[width=0.235\textwidth]{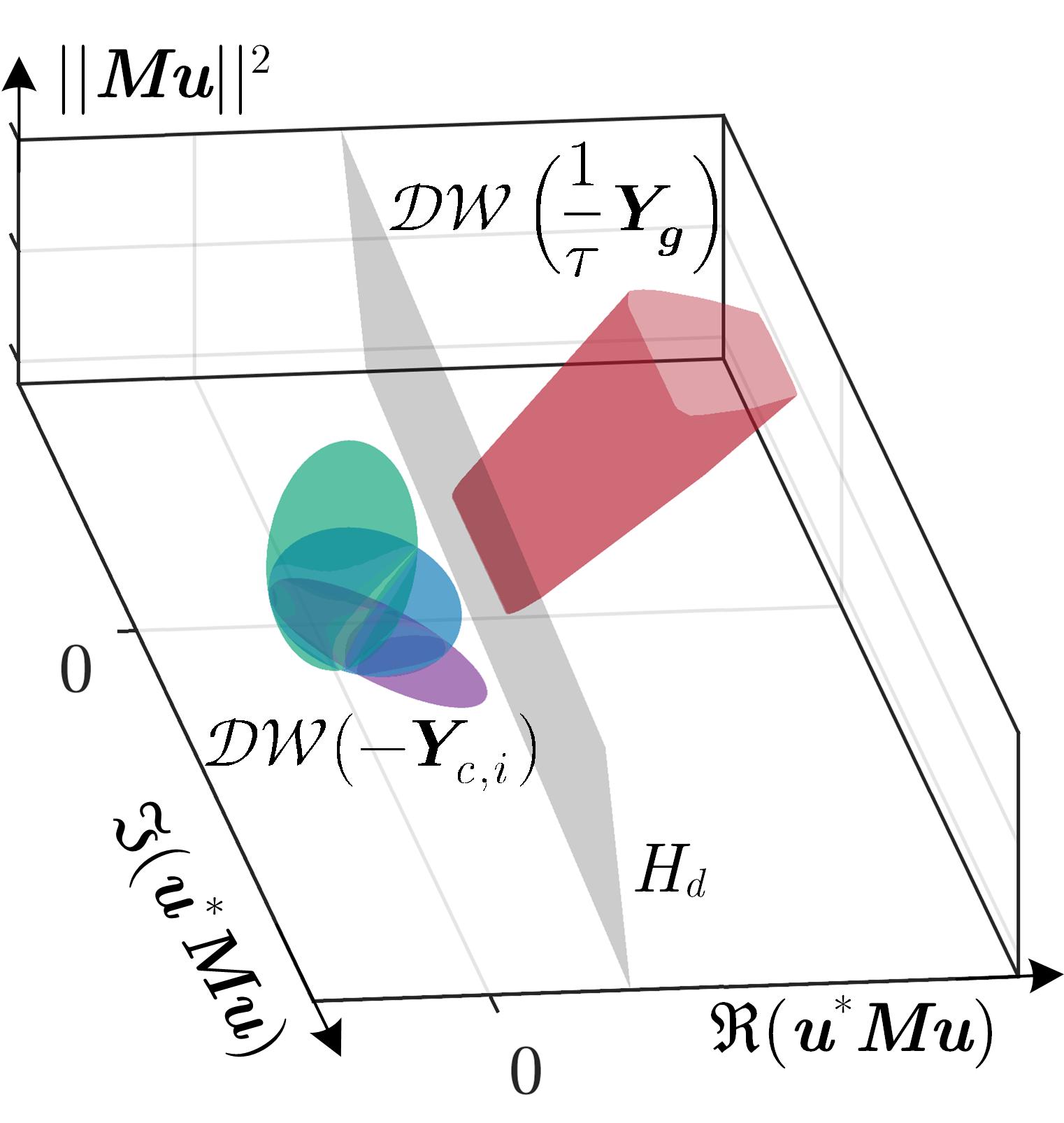}
        \label{fig:existing_conditions_dw}}
    \hfill
    \subfloat[]{%
        \includegraphics[width=0.235\textwidth]{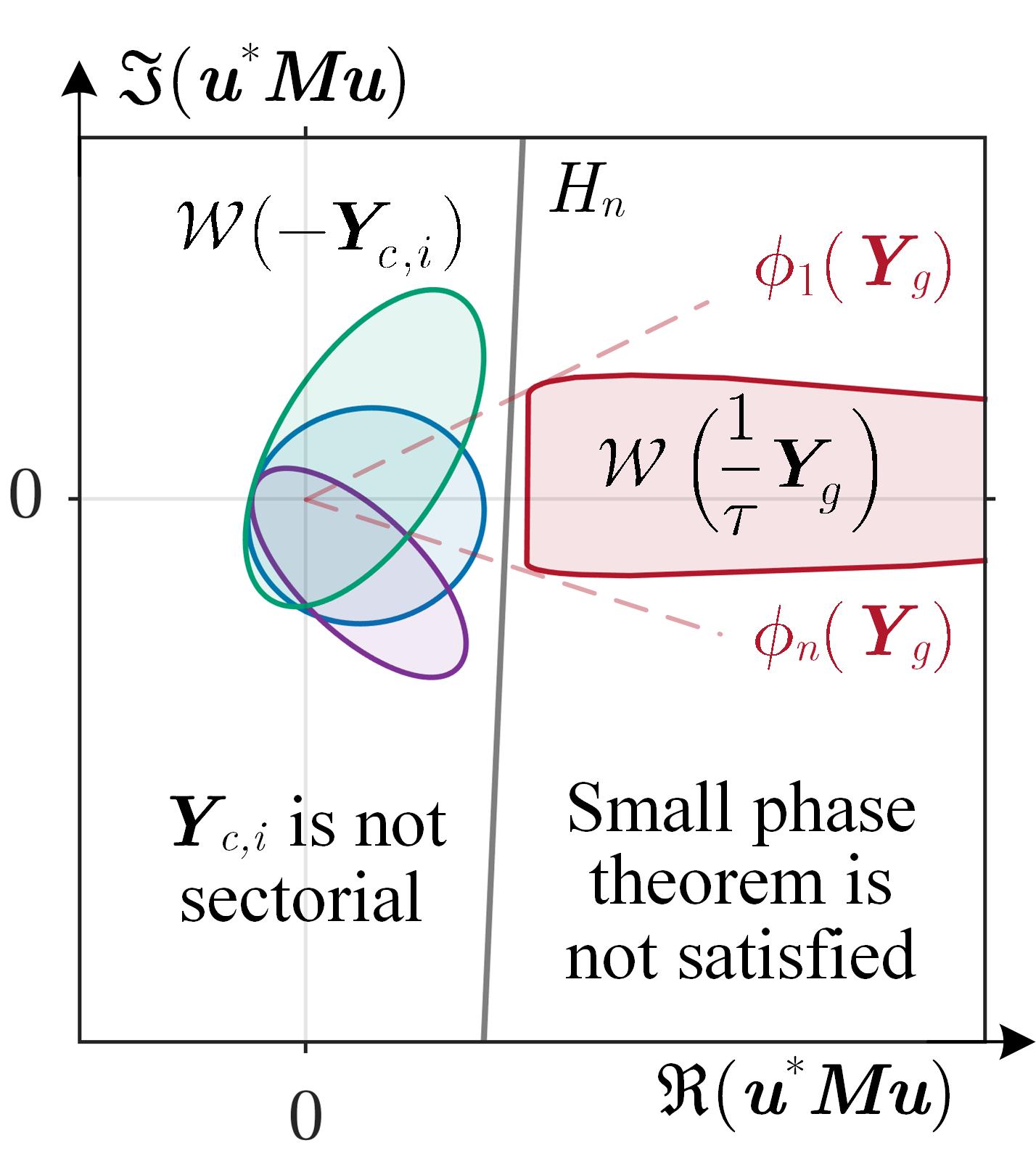}
        \label{fig:existing_conditions_nr}}
    \hfill
    \subfloat[]{%
        \includegraphics[width=0.235\textwidth]{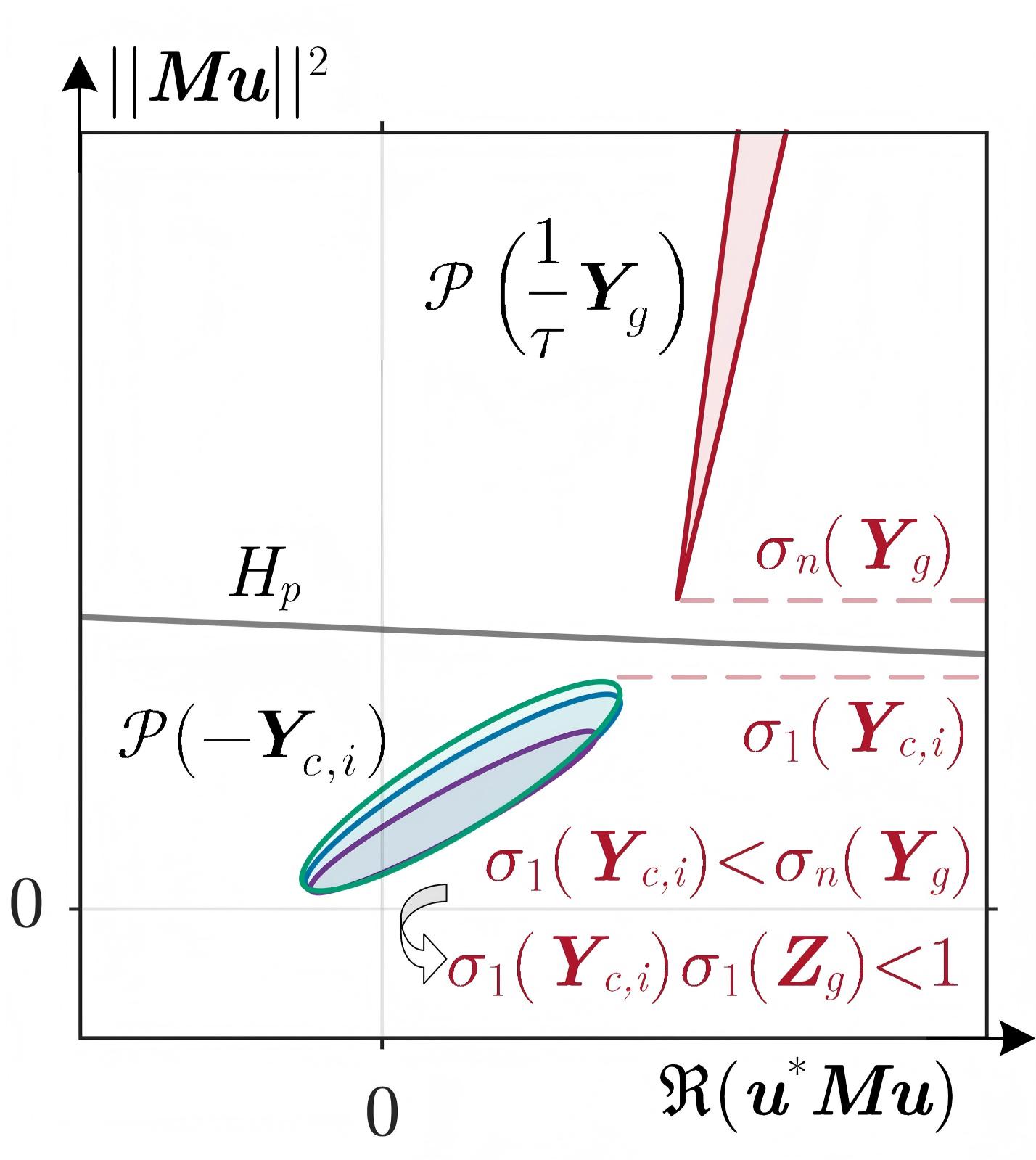}
        \label{fig:existing_conditions_xz}}
    \hfill
    \subfloat[]{%
        \includegraphics[width=0.235\textwidth]{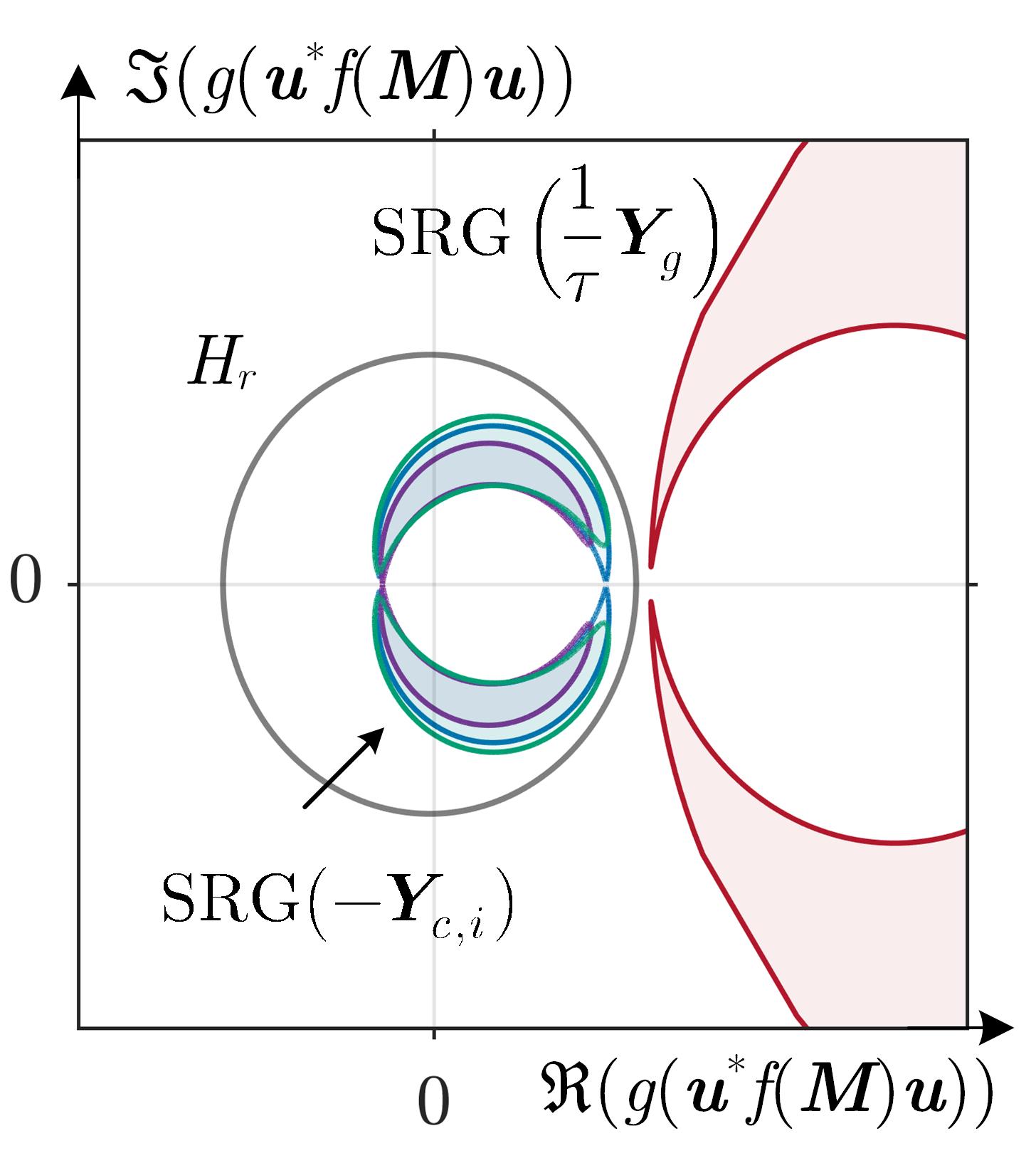}
        \label{fig:existing_conditions_srg}}
    \vspace{-1mm}
    \caption{Geometric interpretations of existing stability conditions for the converter-side and network-side sets at a fixed pair $(\omega,\tau)$. 
    (a) DW shell representation. 
    (b) Numerical range representation. 
    (c) $x$--$z$ graph representation. 
    (d) SRG representation.}
    \label{fig:existing_conditions}
    \vspace{-3mm}
\end{figure*}

\subsection{Discussion on the Decentralization of Stability Conditions}\label{sec:Decentralization_discussion}

The FQC formulation provides a direct route to decentralized geometric stability conditions. Since the converter admittance matrix $\bm{Y}_c(s)$ in~\eqref{Equ_converter_model} is block diagonal, the inequality~\eqref{IQC_DW_Yc} decomposes into local inequalities for each converter $\bm{Y}_{c,i}(s)$
\begin{equation}\label{IQC_DW_Yc_decentralized}
\begin{aligned}
& a\bm{Y}_{c,i}(j\omega)^*\bm{Y}_{c,i}(j\omega)-b\bm{Y}_{c,i}(j\omega)\\
& -b^*\bm{Y}_{c,i}(j\omega)^*-c\bm{I}\preceq -\epsilon\bm{I}.
\end{aligned}
\end{equation}

Taking the DW shell condition as an example, for each fixed pair \((\omega,\tau)\), these local inequalities use the same scalar multipliers \(a\in\mathbb{R}\), \(b\in\mathbb{C}\), and \(c\in\mathbb{R}\) for all converters. Together with the network-side inequality~\eqref{IQC_DW_Zg}, they recover the original centralized FQC-based condition. Thus, each converter can be certified locally under a multiplier shared by all converters.

Geometrically, each converter contributes its own DW shell, and a common hyperplane \(H_d\) must separate the inverse network shell from all converter shells. This enables decentralized verification and ensures \(\mathcal{DW}(-\bm{Y}_{c,i}(j\omega))\cap\mathcal{DW}^{-1}(\tau\bm{Z}_g(j\omega))=\emptyset\) for all \(i\), as shown in Fig.~\ref{fig:existing_conditions}\subref{fig:existing_conditions_dw}. However, the fact that each individual converter DW shell is disjoint from the network DW shell does not guarantee such a common hyperplane. Indeed, since $\mathcal{DW}(-\bm{Y}_c(j\omega))$ is the convex hull of the union of $\mathcal{DW}(-\bm{Y}_{c,i}(j\omega))$ for all $i=1,\ldots,n$, this aggregate convex hull may intersect the network DW shell even when every individual shell is disjoint from it. Therefore, a decentralized DW shell stability condition requires not merely the disjointness between each individual converter shell and the network shell, but the existence of a common hyperplane that separates these converter shells from the network shell.

As shown in Fig.~\ref{fig:existing_conditions}\subref{fig:existing_conditions_nr}--\subref{fig:existing_conditions_srg}, the above discussion can be extended in a similar manner to the numerical range, $x$-$z$ graph, and SRG separation conditions in Proposition~\ref{pro:exist_stability_conditions}, since they can also be regarded as specializations of the FQC framework. For these geometric conditions, the common multiplier triplet determines the separating boundary and must satisfy the network-side FQC and all local converter-side FQCs. Updating a converter may invalidate its local inequality \textit{under the existing multiplier}, so the corresponding boundary may no longer separate its set. In that case, a new feasible triplet \((a,b,c)\), and corresponding separating boundary, must be sought to restore the decentralized certificate.

A simpler decentralized implementation is available for the \(x\)-\(z\) graph condition when the network contains only transmission lines with the same \(R/X\) ratio. After incorporating the common network dynamics into the converter models~\cite{ref_Gain_Phase_2}, the network admittance can be represented by a positive definite matrix \(\bm{Y}_g=\bm{Y}_g^*\succ0\). The separating line \(H_p\) can be determined analytically by the extreme eigenvalues of \(\frac{1}{\tau}\bm{Y}_g\), independently of the converter dynamics, so replacing or updating a converter requires no new common hyperplane search. Within the FQC framework of~\eqref{eq_IQC_all}, this property corresponds to choosing \(a\), \(b\), \(c\) based solely on the network parameters and \(\tau\). With these multipliers fixed, only the local FQC inequality of the updated converter needs to be re-evaluated. The following proposition formalizes this network-parameterized certificate.

%Thus, if the network remains unchanged, only the local QC inequality of a replaced or updated converter needs to be re-evaluated. The scalar multipliers \(a\), \(b\), and \(c\) need not be recomputed, and neither the network-side inequality nor the local inequalities of the unchanged converters need to be re-evaluated. The following proposition gives the resulting multipliers and local stability condition explicitly.

\begin{Proposition}[Network-parameterized decentralized \(x\)--\(z\) graph stability condition]
\label{pro:scalable_xz_envelope}
Suppose that the open-loop systems $\bm{Y}_c(s)$ and $\bm{Z}_g(s)$ are stable and $\bm{Y}_g=\bm{Y}_g^*\succ0$. Let $\lambda_-:=\lambda_{\min}(\bm{Y}_g)$, $\lambda_+:=\lambda_{\max}(\bm{Y}_g)$, $S:=\lambda_-+\lambda_+$, and $P:=\lambda_-\lambda_+$. Then, the negative feedback interconnection of $\bm{Y}_c(s)$ and $\bm{Z}_g(s)$ is stable if, for all $\omega\in[0,\infty)$, $\tau\in(0,1]$, and $i=1,\ldots,n$, there exists $\epsilon>0$ such that
\begin{equation}
\label{Equ_local_xz_envelope}
\bm{Y}_{c,i}(j\omega)^*\bm{Y}_{c,i}(j\omega)
+\frac{S}{2\tau}\left(\bm{Y}_{c,i}(j\omega)+\bm{Y}_{c,i}(j\omega)^*\right)+\frac{P}{\tau^2}\bm{I}\succeq\epsilon\bm{I}.
\end{equation}
\end{Proposition}

\begin{proof}
The proof is given in Appendix~\ref{app:proof_2}.
\end{proof}

As illustrated in Fig.~\ref{fig:XZ_decentralized}, for each \(\tau\in(0,1]\), the $x$--$z$ graph of the Hermitian network matrix $\bm Y_g/\tau$ is bounded above by $L_\tau(x)=Sx/\tau-P/\tau^2$. The choice $a=-1$, $b=S/(2\tau)$, and $c=P/\tau^2$ in Appendix~\ref{app:proof_2} gives this network-side boundary, while~\eqref{Equ_local_xz_envelope} places each converter graph strictly above it. Since $S$ and $P$ depend only on the network, updating a converter requires rechecking only its local FQC inequality.

\begin{figure}
\centering
\includegraphics[width=0.9\columnwidth]{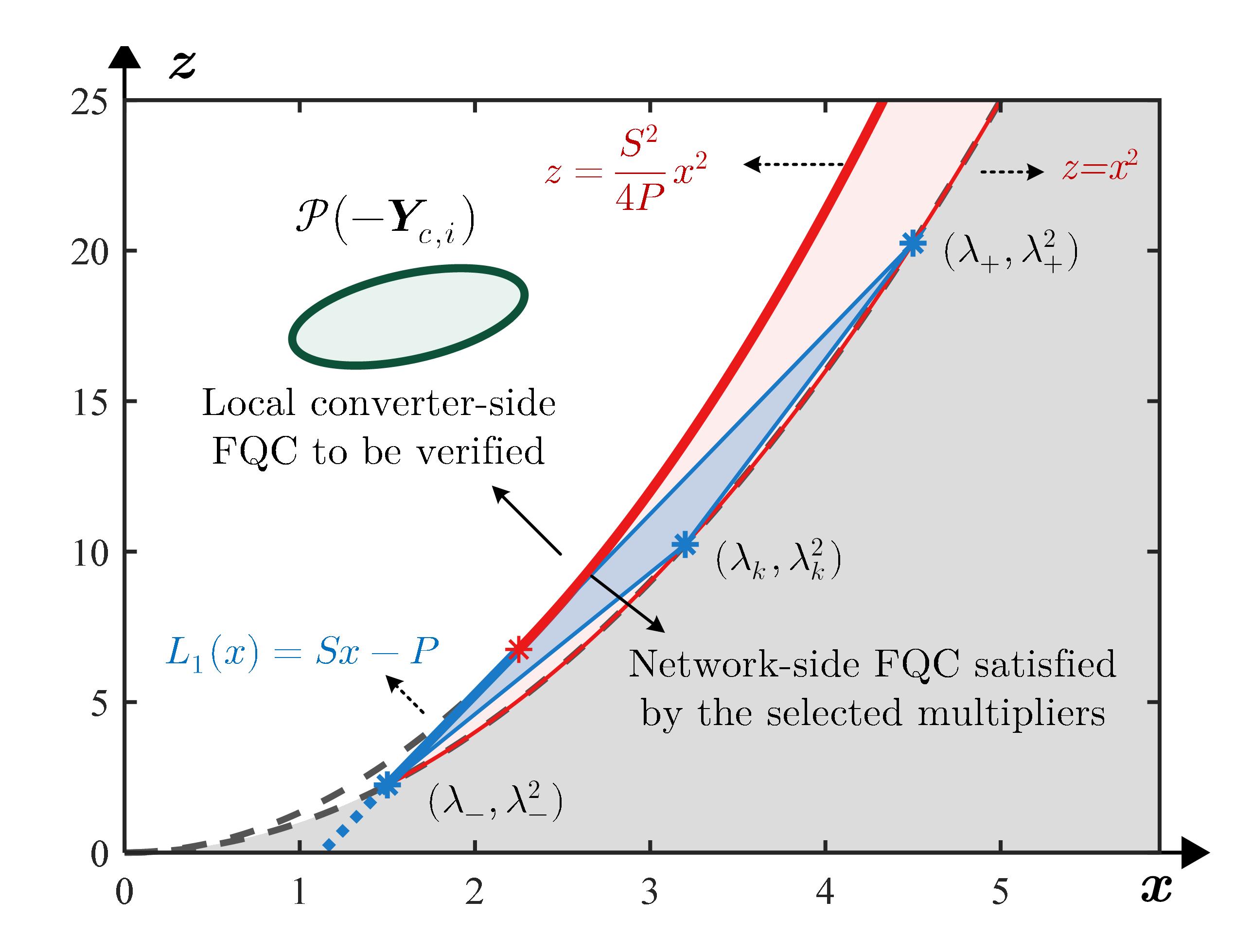}
\vspace{-6mm}
\caption{Geometric interpretation of the network-parameterized decentralized \(x\)--\(z\) graph condition. The blue polygon represents \(\mathcal P(\bm Y_g)\) at \(\tau=1\), and the red-shaded region is swept by \(\mathcal P(\bm Y_g/\tau)\) over \(\tau\in(0,1]\). The green ellipse represents the local converter graph \(\mathcal P(-\bm Y_{c,i})\). The gray region \(z<x^2\) is unattainable by any \(x\)--\(z\) graph. For each \(\tau\), \(H_p\) coincides with the network upper boundary \(L_\tau\), so the selected multipliers satisfy the network-side FQC, and only the local converter-side condition requires verification.}
\label{fig:XZ_decentralized}
\end{figure}

In contrast to the geometric separation conditions above, the gain and phase conditions in Proposition~\ref{pro:exist_stability_conditions} can be decentralized without constructing a common separating hyperplane. Their decentralization can instead be assessed using device-level gain scalars and phase areas, leading to the following decentralized mixed gain-phase condition~\cite{ref_Gain_Phase_2}. 

\begin{Proposition}[Decentralized mixed gain-phase condition]
\label{pro:decentralized_gain_phase}
Suppose that the open-loop systems $\bm{Y}_c(s)$ and $\bm{Z}_g(s)$ are stable. Then, the negative feedback interconnection of $\bm{Y}_c(s)$ and $\bm{Z}_g(s)$ is stable if, for all $\omega\in[0,\infty)$, either

i) the decentralized gain condition, i.e.,
\begin{equation}
\max_{i}\sigma_1(\bm{Y}_{c,i}(j\omega))<\sigma_{2n}(\bm{Y}_g(j\omega)),
\label{Equ_decentralized_gain}
\end{equation}
holds, or

ii) the decentralized phase condition, i.e.,
\begin{equation}
\left\{
\begin{aligned}
a)\;& \bm{Y}_{c,i}(j\omega)\ \mathrm{and}\ \bm{Z}_g(j\omega)\ \mathrm{are\ sectorial},\\
b)\;& \max_{i}\phi_1(\bm{Y}_{c,i}(j\omega))
-\min_{i}\phi_2(\bm{Y}_{c,i}(j\omega))<\pi,\\
c)\;& \max_{i}\phi_1(\bm{Y}_{c,i}(j\omega))<\pi-\phi_1(\bm{Z}_g(j\omega)),\\
d)\;& \min_{i}\phi_2(\bm{Y}_{c,i}(j\omega))>-\pi-\phi_{2n}(\bm{Z}_g(j\omega)),
\end{aligned}
\right.
\label{Equ_decentralized_phase}
\end{equation}
holds.
\end{Proposition}
\begin{proof}
The self-contained proof is given in Appendix~\ref{app:proof_5}~\cite{ref_Gain_Phase_2}.
\end{proof}

At each frequency, the condition in Proposition~\ref{pro:decentralized_gain_phase} can be verified either by comparing every converter gain \(\sigma_1(\bm{Y}_{c,i}(j\omega))\) with the network gain \(\sigma_{2n}(\bm{Y}_g(j\omega))\), or by checking that every converter phase area $[\phi_2(\bm{Y}_{c,i}(j\omega)),\phi_1(\bm{Y}_{c,i}(j\omega))]$ is contained within the network phase area $[-\pi-\phi_{2n}(\bm{Z}_{g}(j\omega)),\pi-\phi_1(\bm{Z}_{g}(j\omega))]$, and the union of all converters' phase areas has a width less than $\pi$. Note that if a matrix is not sectorial, we consider its phase area to be $(-\infty,+\infty)$. Both tests use converter-level gain or phase information and require no common separating hyperplane.

\section{Optimization-Based Scalable Stability Assessment via FQCs}\label{sec:screening-diagnosis}

This section first develops a full-multiplier scalable stability condition within the FQC framework, which provides additional freedom beyond scalar-multiplier geometric conditions and is therefore even less conservative than the DW shell stability condition. We formulate the frequency-wise optimization to identify potentially unstable frequency ranges and problematic devices. Finally, we combine this optimization with the mixed gain-phase condition to form a hierarchical screening-diagnosis procedure for efficient stability assessment.

\subsection{Full-Multiplier Scalable Stability Condition}

Section~\ref{sec:existing_decentralized_stability_conditions} shows that existing stability conditions can be interpreted as special cases of the FQC framework obtained by imposing scalar structures on the multipliers. Although these scalar multipliers provide transparent geometric interpretations, their restricted form may introduce additional conservatism. To reduce conservatism, this scalar structure can be relaxed by introducing the full multiplier, where the scalar parameters $a$, $b$, and $c$ are replaced by $2\times2$ matrices $\bm{A}_0$, $\bm{B}_0$, and $\bm{C}_0$. The corresponding multiplier matrices in \eqref{IQC_3} and \eqref{IQC_4} are developed as
\begin{equation*}
%\label{Equ_full_multiplier}
\bm{A}=\bm{I}_n\otimes \bm{A}_0,\quad
\bm{B}=\bm{I}_n\otimes \bm{B}_0,\quad
\bm{C}=\bm{I}_n\otimes \bm{C}_0,
\end{equation*}
where $\bm{A}_0=\bm{A}_0^*\in\mathbb{C}^{2\times 2}$, $\bm{C}_0=\bm{C}_0^*\in\mathbb{C}^{2\times 2}$, and $\bm{B}_0\in\mathbb{C}^{2\times 2}$. 

Here and hereafter, the frequency argument $j\omega$ is omitted when no ambiguity arises. The network-side condition is expressed as
\begin{equation}
\label{Equ_Hg_tau}
\begin{aligned}
\bm{H}_g(\omega,\tau)&=\bm{Y}_g^*(\bm{I}_n\otimes\bm{A}_0)\bm{Y}_g
+\tau(\bm{I}_n\otimes\bm{B}_0)\bm{Y}_g\\
&\quad+\tau\bm{Y}_g^*(\bm{I}_n\otimes\bm{B}_0^*)-\tau^2(\bm{I}_n\otimes\bm{C}_0)\succeq0 .
\end{aligned}
\end{equation}

The corresponding converter-side condition is expressed as
\begin{equation}\label{Equ_diag_Wi}
\text{diag}(\bm{W}_1,\ldots,\bm{W}_n)\succeq\epsilon\bm{I},
\end{equation}
where $\bm{W}_i=-\bm{Y}_{c,i}^*\bm{A}_0\bm{Y}_{c,i}
+\bm{B}_0\bm{Y}_{c,i}+\bm{Y}_{c,i}^*\bm{B}_0^*+\bm{C}_0$.

These conditions yield the following stability result.

\begin{Proposition}[Full-multiplier scalable stability condition]
\label{pro:shared_full_multiplier}
Suppose that the open-loop systems $\bm{Y}_c(s)$ and $\bm{Z}_g(s)$ are stable. Then, the negative feedback interconnection of $\bm{Y}_c(s)$ and $\bm{Z}_g(s)$ is stable if for all $\omega\in[0,\infty)$ and $\tau\in(0,1]$, there exist $\bm{A}_0=\bm{A}_0^*$, $\bm{C}_0=\bm{C}_0^*$, $\bm{B}_0$ and $\epsilon>0$ such that $\bm{H}_g(\omega,\tau)\succeq0$ and $\bm{W}_i\succeq\epsilon\bm{I}$ for all $i=1,\ldots,n$.
\end{Proposition}
\begin{proof}
The proof is given in Appendix~\ref{app:proof_3}.
\end{proof} 

In Proposition~\ref{pro:shared_full_multiplier}, the parameter $\tau$ needs to be scanned over the entire interval $(0,1]$. This scanning process can be computationally burdensome, especially when the frequency grid is dense. To mitigate this issue, the following proposition provides a simplified endpoint stability condition that only requires checking the network-side inequality at $\tau=1$.

\begin{Proposition}[Endpoint full-multiplier scalable stability condition]
\label{pro:endpoint_full_multiplier}
Suppose that the open-loop systems $\bm{Y}_c(s)$ and $\bm{Z}_g(s)$ are stable. Then, the negative feedback interconnection of $\bm{Y}_c(s)$ and $\bm{Z}_g(s)$ is stable if for all $\omega\in[0,\infty)$, there exist $\bm{A}_0=\bm{A}_0^*\succeq0$, $\bm{C}_0=\bm{C}_0^*\succeq0$, $\bm{B}_0$ and $\epsilon>0$ such that $\bm{H}_g(\omega,1)\succeq0$ and $\bm{W}_i\succeq\epsilon\bm{I}$ for all $i=1,\ldots,n$.
\end{Proposition}

\begin{proof}
The proof is given in Appendix~\ref{app:proof_4}.
\end{proof}

Proposition~\ref{pro:shared_full_multiplier} and~\ref{pro:endpoint_full_multiplier} enable scalable stability analysis in multi-converter systems. The existence of the common full multiplier $(\bm{A}_0,\bm{B}_0,\bm{C}_0)$ that simultaneously certifies the network-side inequality and local converter-side inequalities ensures the interconnected system stability. Compared with the scalar multipliers in~\eqref{IQC_DW_Zg} and~\eqref{IQC_DW_Yc_decentralized}, this full-multiplier structure offers additional degrees of freedom for reducing conservatism, while retaining separate converter-side inequalities.

For a fixed common full multiplier \((\bm{A}_0,\bm{B}_0,\bm{C}_0)\), the local converter-side inequalities in Proposition~\ref{pro:shared_full_multiplier} and Proposition~\ref{pro:endpoint_full_multiplier} define the FQC stability margin \(\lambda_{\min}(\bm{W}_i)\) for the \(i\)-th converter. These margins support converter-level comparisons and the optimization problem developed below for diagnosis.

\subsection{Optimization-Based Screening and Diagnosis}

Based on Proposition~\ref{pro:endpoint_full_multiplier}, we formulate a sequential frequency-wise optimization procedure to determine whether a full multiplier $(\bm{A}_0,\bm{B}_0,\bm{C}_0)$ and a positive margin $\epsilon$ can be found to satisfy the stability conditions $\bm{H}_g(\omega,1)\succeq0$ and $\bm{W}_i\succeq\epsilon\bm{I}$. The first stage (Stage-I) determines whether these FQC-based conditions admit a positive common margin, thereby certifying system stability. If a positive common margin cannot be obtained at some frequency points, the second stage (Stage-II) evaluates the individual converter margins to identify weak-converter candidates associated with potential instability risks. Both stages are convex optimization problems with linear matrix inequality (LMI) constraints, which can be efficiently solved using semidefinite programming (SDP) techniques. They provide a numerical FQC certification margin and converter-level diagnostic indices directly from converter and network matrices, without constructing or visually inspecting graphical sets. The formulations are given below.

Since the FQC inequalities in \eqref{Equ_Hg_tau} and \eqref{Equ_diag_Wi} are homogeneous in the multiplier, any feasible triplet $(\bm{A}_0,\bm{B}_0,\bm{C}_0)$ can be multiplied by a positive scalar without changing feasibility, while the associated margins are scaled by the same factor. A normalization is therefore needed to prevent an ill-posed margin maximization and to place the optimized margins on a fixed scale. For a fixed $\omega$, the normalization constraint can be expressed as
\begin{equation}
\label{Equ_normalization}
\mathrm{tr}(\bm{A}_0)+\mathrm{tr}(\bm{C}_0)=1,\quad
\|\bm{P}_0\|_F\le 1,
\end{equation}
where $\bm{P}_0=\bigl[\begin{smallmatrix}\bm{A}_0 & \bm{B}_0^*\\ \bm{B}_0 & \bm{C}_0\end{smallmatrix}\bigr]$; $\mathrm{tr}(\cdot)$ denotes the trace of a matrix; $\|\cdot\|_F$ denotes the Frobenius norm of a matrix. 

%The first constraint in \eqref{Equ_normalization} fixes the scale of the multiplier by normalizing the sum of the traces of $\bm{A}_0$ and $\bm{C}_0$ to be 1, while the second constraint limits the Frobenius norm of $\bm{P}_0$ to be no larger than a prescribed constant $R>0$. This second constraint is introduced to prevent numerical issues that may arise from excessively large multipliers, which can lead to ill-conditioned optimization problems. By bounding the norm of $\bm{P}_0$, we ensure that the optimization remains well-posed and numerically stable, while still allowing for sufficient flexibility in the choice of the multiplier. The specific value of $R$ can be chosen based on numerical experimentation or domain knowledge, and it serves as a regularization parameter that balances between conservatism and numerical tractability in the optimization process.

%with a prescribed constant $R>0$. 

For each fixed $\omega$, the Stage-I optimization maximizes the common FQC margin $\epsilon$ and is formulated as
\begin{subequations}
\label{Opt_stage1}
\begin{align}
\epsilon^{\star}(\omega)=
&\max_{\bm{A}_0,\bm{B}_0,\bm{C}_0,\epsilon}\ \epsilon \label{Opt_stage1_a}\\
\mathrm{s.t.}\quad
&\bm{W}_i(\omega;\bm{A}_0,\bm{B}_0,\bm{C}_0)\succeq \epsilon \bm{I}_2,\quad i=1,\ldots,n, \label{Opt_stage1_b}\\
&\bm{H}_g(\omega,1;\bm{A}_0,\bm{B}_0,\bm{C}_0)\succeq0, \label{Opt_stage1_c}\\
&\bm{A}_0=\bm{A}_0^*\succeq0,\quad \bm{C}_0=\bm{C}_0^*\succeq0, \label{Opt_stage1_d}\\
&\mathrm{tr}(\bm{A}_0)+\mathrm{tr}(\bm{C}_0)=1,\quad
\|\bm{P}_0\|_F\le 1. \label{Opt_stage1_e}
\end{align}
\end{subequations}

A frequency point \(\omega\) is classified as certified if \(\epsilon^\star(\omega)>0\), in which case all converters admit a strictly positive local margin under the same multiplier. Otherwise, it is classified as uncertified. If certified points can be obtained for every $\omega\in[0,\infty)$, then the scalable stability condition in Proposition~\ref{pro:endpoint_full_multiplier} is satisfied, and the closed-loop interconnected system is stable.
It should be noted that, if the uncertified points are obtained at some frequencies, this only indicates that the proposed scalable sufficient condition is not satisfied. It does not by itself imply that the closed-loop system is unstable, but it indicates a potential stability risk in the corresponding frequency range.

Stage-I maximizes the worst local FQC margin rather than the distribution of margins across converters. To obtain device-level information, Stage-II is applied to frequency points that remain uncertified after Stage-I. It assigns each converter a local FQC margin \(\epsilon_i\) and maximizes their sum as follows
\begin{subequations}\label{Opt_stage2}
\begin{align}
&\max_{\bm{A}_0,\bm{B}_0,\bm{C}_0,\epsilon_1,\ldots,\epsilon_n}\ \sum_{i=1}^{n}\epsilon_i \label{Opt_stage2_a}\\
\mathrm{s.t.}\quad
&\bm{W}_i(\omega;\bm{A}_0,\bm{B}_0,\bm{C}_0)\succeq \epsilon_i\bm{I}_2,\quad i=1,\ldots,n, \label{Opt_stage2_b}\\
&\bm{H}_g(\omega,1;\bm{A}_0,\bm{B}_0,\bm{C}_0)\succeq0, \label{Opt_stage2_c}\\
&\bm{A}_0=\bm{A}_0^*\succeq0,\quad \bm{C}_0=\bm{C}_0^*\succeq0, \label{Opt_stage2_d}\\
&\mathrm{tr}(\bm{A}_0)+\mathrm{tr}(\bm{C}_0)=1,\quad
\|\bm{P}_0\|_F\le 1. \label{Opt_stage2_e}
\end{align}
\end{subequations}

%Since the stability condition must hold along the whole path $\tau\in[0,1]$, the device-level diagnostic index at a fixed frequency is defined by the worst local margin over the checked path parameters. 

After solving \eqref{Opt_stage2}, the diagnostic index of the $i$-th converter is computed as
\begin{equation}
\label{Equ_gamma_i}
\gamma_i(\omega)=\lambda_{\min}\!\left(\bm{W}_i(\omega;\bm{A}_0^{\star}(\omega),\bm{B}_0^{\star}(\omega),\bm{C}_0^{\star}(\omega))\right),
\end{equation}
where $(\bm{A}_0^{\star},\bm{B}_0^{\star},\bm{C}_0^{\star})$ denotes the optimal full-multiplier triplet.

%Stage-II should be interpreted as a diagnostic refinement rather than an additional stability theorem. It provides converter-level indices for comparing all local margins when Stage-I fails to produce a stability certificate. A smaller $\gamma_i(\omega)$ indicates a weaker local margin under the optimized multiplier $(\bm{A}_0^{\star},\bm{B}_0^{\star},\bm{C}_0^{\star})$. In particular, if $\gamma_i(\omega)\leq0$, then the local converter-side QC inequality is not satisfied for converter $i$. The converter $\bm{Y}_{c,i}(j\omega)$ is therefore identified as a weak converter associated with a potential small-signal stability risk in the corresponding frequency range.

The index $\gamma_i(\omega)$ quantifies the local FQC margin of converter $i$ under the selected optimal multiplier $(\bm{A}_0^{\star},\bm{B}_0^{\star},\bm{C}_0^{\star})$. A converter with $\gamma_i(\omega)<0$ is identified as a weak device in the corresponding frequency range. Although the optimized sum may be positive, it does not constitute a stability certificate because some local FQC inequalities may remain unsatisfied.

\subsection{Scalable Stability Assessment Procedure}

The above optimization provides scalable stability certification and device-level diagnosis, but solving an SDP at every frequency is computationally demanding. We therefore use the mixed gain-phase condition in Proposition~\ref{pro:decentralized_gain_phase} for preliminary screening. Compared with the general geometric separation conditions in Proposition~\ref{pro:exist_stability_conditions}, it enables decentralized verification through direct comparisons of converter-level scalar gains or phase areas with network-side bounds, without searching for a common separating hyperplane or comparing two- or three-dimensional graphical sets at each $(\omega,\tau)$. Although more conservative, this test is inexpensive to evaluate, and its gain and phase plots provide a direct visual indication of which frequency ranges are certified.

\begin{figure}
\centering
\includegraphics[width=\columnwidth]{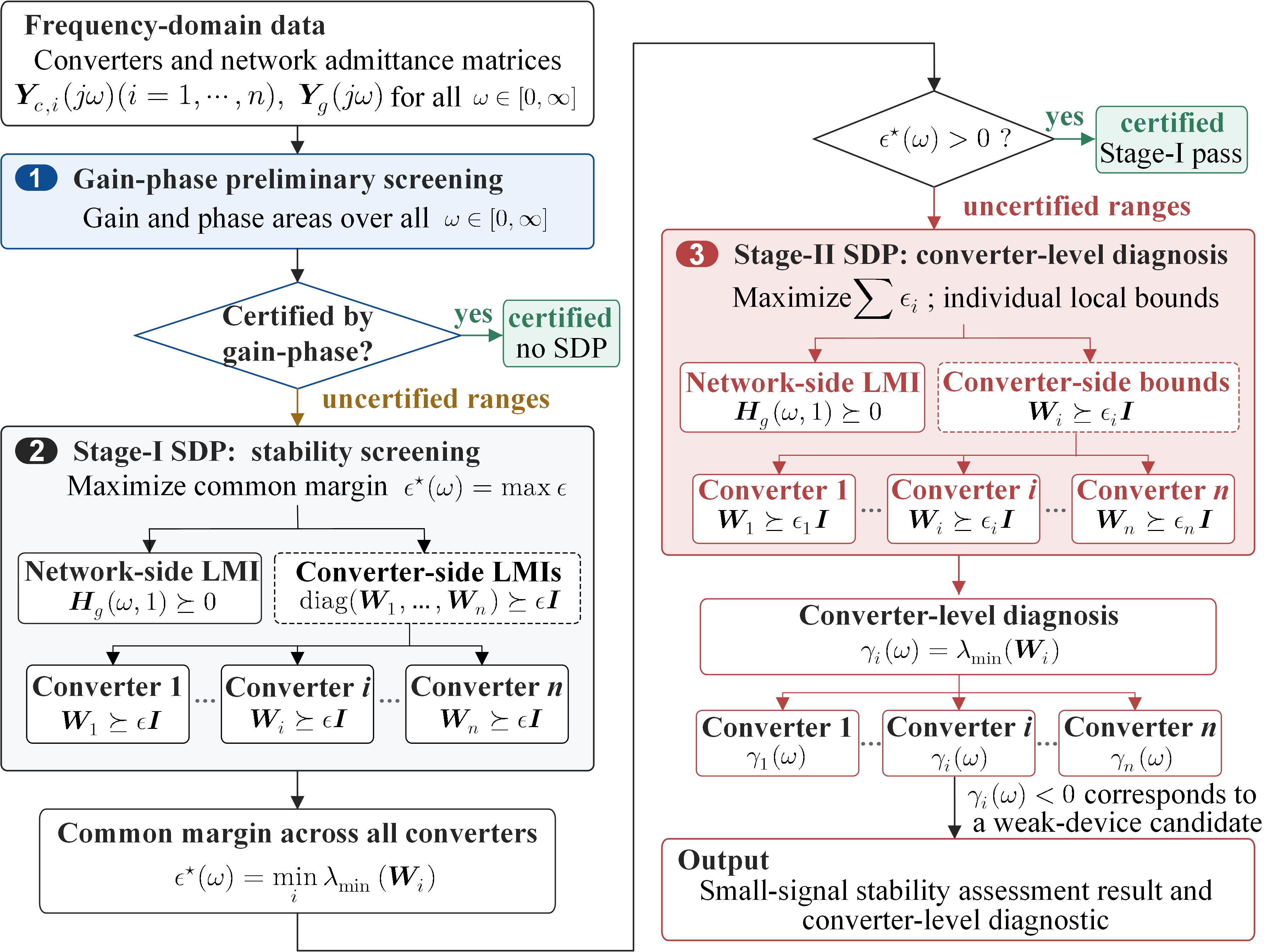}
\vspace{-6mm}
\caption{Hierarchical screening-diagnosis procedure for scalable small-signal stability assessment of multi-converter systems.}
\label{fig:flow_chart}
\end{figure}

Hence, the proposed scalable stability assessment method is a hierarchical screening-diagnosis procedure, as shown in Fig.~\ref{fig:flow_chart}. The decentralized mixed gain-phase test is first applied over the full frequency range. Stage-I is then solved only for ranges left uncertified by this test. When Stage-I still has uncertified frequency ranges, Stage-II is solved to compute the device-level indices $\gamma_i(\omega)$ and identify weak-device candidates. In this way, the high-cost optimization is concentrated on the frequency ranges where the simpler gain-phase test fails, while retaining the ability to reduce conservatism and provide converter-level diagnostic information.

\section{Case Studies}\label{sec:case}

This section validates the proposed scalable stability assessment method on a modified 68-bus system shown in Fig.~\ref{fig:IEEE_68}. The system contains four GFM converters connected to Buses 1--4, nine GFL converters connected to Buses 5--13, and two SGs connected to Buses 14 and 15. Bus 16 is modeled as an infinite bus, which represents a remote external grid and provides the voltage-angle reference. The transmission lines are represented by Pi models, including the series resistance, inductance, and shunt capacitances of each branch. Details of the line parameters, load configuration, and converter and generator models can be found in the 68-bus case study in~\cite{ref_Gain_Phase_2}. The PLL bandwidths of GFL converters 5--13 are set to 10, 15, 20, 25, 30, 35, 40, 60, 70 rad/s, respectively, ensuring distinct dynamic characteristics for all GFL converters. Following~\cite{ref_Gain_Phase_2}, the GFM converters and SGs are treated as grid-supporting devices and incorporated into the network side. The GFM converters and SGs are assumed to be stable when they are interconnected, which is usually guaranteed by the design of GFM converters and SGs for parallel operation and is not the focus here. The assessment therefore focuses on the interaction between the nine GFL converters and the remaining system.

\begin{figure}[t]
\centering
\includegraphics[width=\columnwidth]{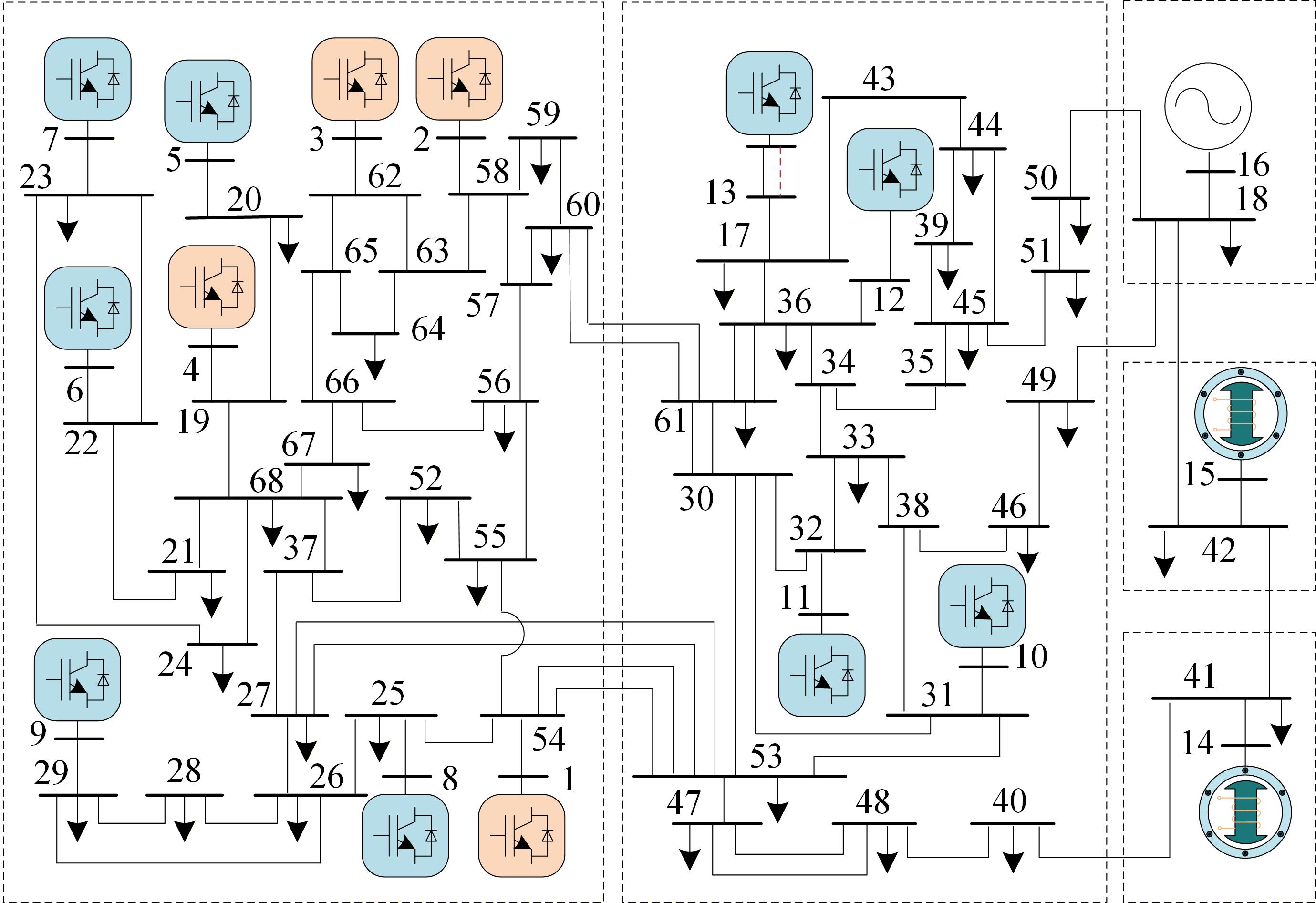}
\vspace{-6mm}
\caption{A modified 68-bus system integrated with 9 GFL converters, 4 GFM converters and 2 SGs, where \textcolor{myblue}{\raisebox{0.2ex}{\rule{0.1\linewidth}{4pt}}} denotes GFL converters, \textcolor{myorange}{\raisebox{0.2ex}{\rule{0.1\linewidth}{4pt}}} denotes GFM converters and \textcolor{mygreen}{\raisebox{0.2ex}{\rule{0.1\linewidth}{4pt}}} denotes SGs. Black arrows denote the loads.}
\label{fig:IEEE_68}
\end{figure}

\subsection{Validation of Full-Multiplier Stability Condition}

We first consider a small-signal stable operating condition of the modified 68-bus system. This case examines whether the proposed full-multiplier condition can certify stability at frequencies for which the mixed gain-phase and scalar geometric conditions do not provide a certificate.

To this end, the mixed gain-phase condition in Proposition~\ref{pro:decentralized_gain_phase} is first applied. The gain and phase areas of the GFL converters are compared with those of the network, as shown in Fig.~\ref{fig:gain_phase_stable}. Above 14.4 Hz, all GFL converters are sectorial, and their phase areas are fully contained within that of the network. Moreover, the union of their phase areas has a width less than $180^\circ$. Below 11.2 Hz, all gains of GFL converters lie below the network gain. Between 11.2 Hz and 14.4 Hz (the gray area in Fig.~\ref{fig:gain_phase_stable}), neither gain nor phase condition is satisfied, and the small-signal stability can not be guaranteed.

\begin{figure}[t]
\centering
\includegraphics[width=\columnwidth]{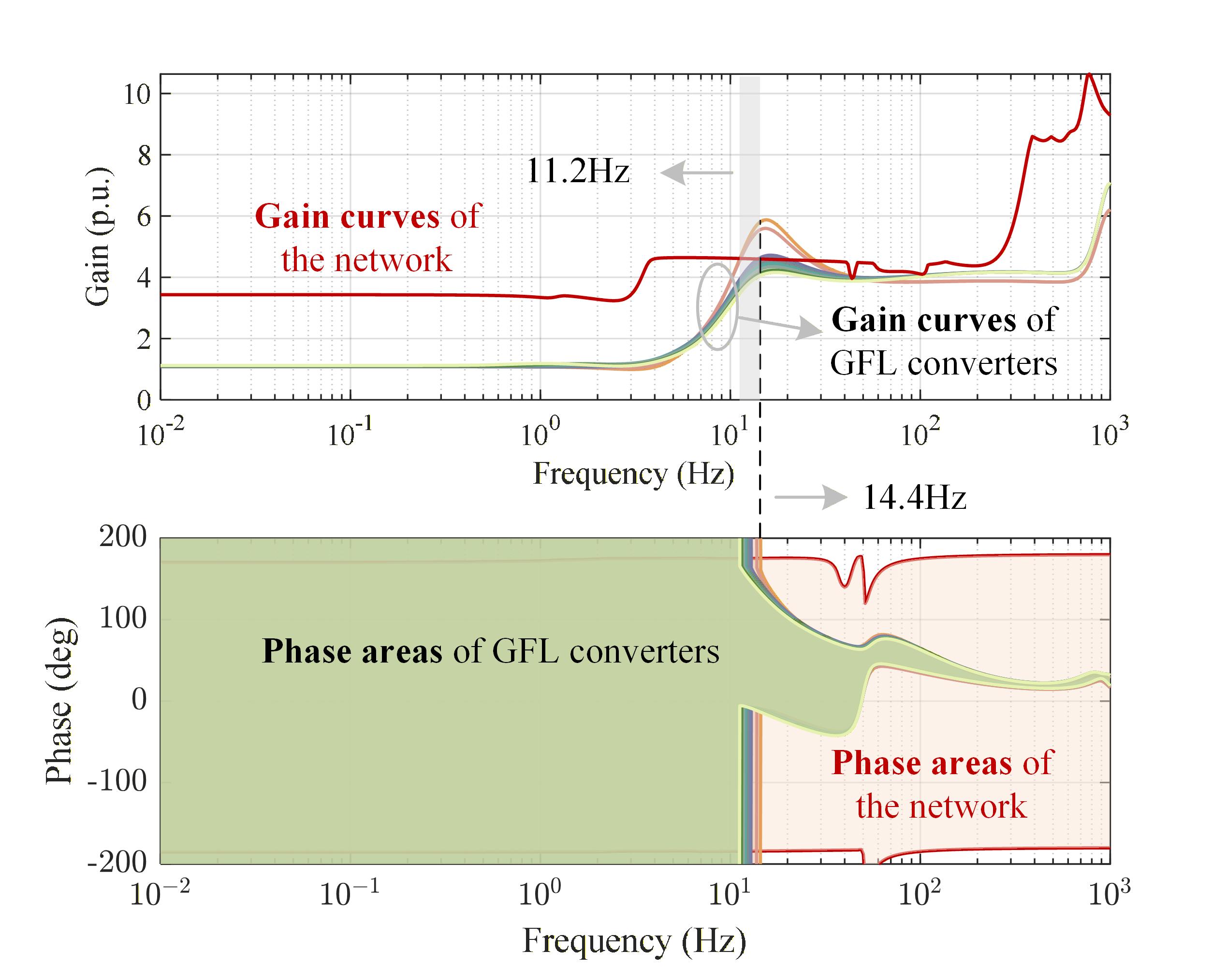}
\vspace{-8mm}
\caption{Gain curves and phase areas of the GFL converters and the network for the stable operating condition.}
\label{fig:gain_phase_stable}
\end{figure}

We then evaluate the DW shell separation condition, which is the least conservative among the scalar geometric conditions in Proposition~\ref{pro:exist_stability_conditions}. 
Due to space limitations, Fig.~\ref{fig:DW_stable} shows only the DW shells of GFL converter 13 and the network over 11.2--14.4~Hz and $\tau\in(0,1]$. An overlap between the converter and network DW shells can be observed from 11.4~Hz to 12.2~Hz.
Therefore, the DW shell separation condition fails to certify stability in this frequency range. Since the numerical range, $x$-$z$ graph, and SRG separation conditions are obtained by imposing additional restrictions on the scalar multiplier, they are more conservative than the DW shell condition and are not shown here due to space limitations.

\begin{figure}[t]
\centering
\includegraphics[width=0.9\columnwidth]{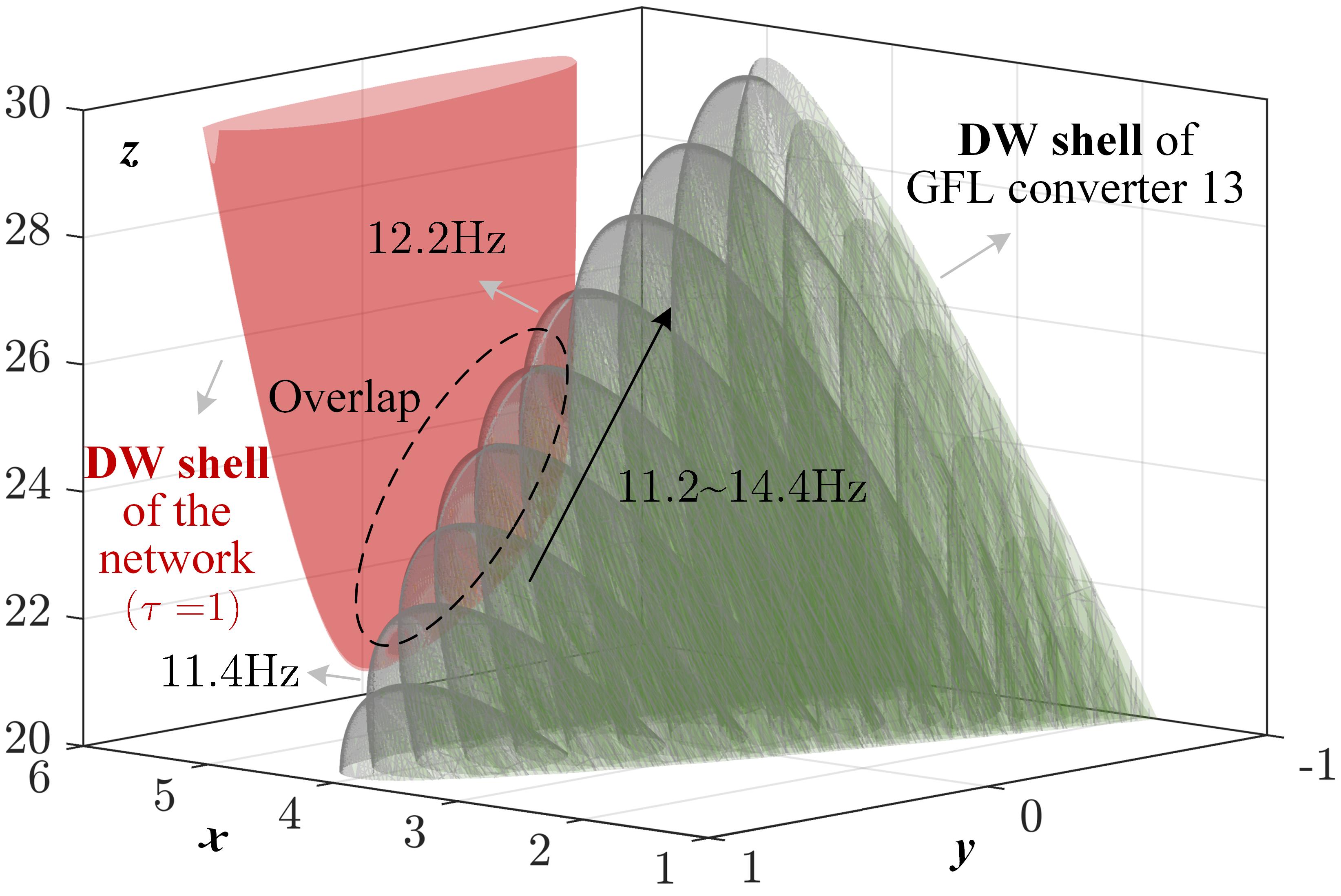}
\vspace{-3mm}
\caption{DW shells of GFL converter 13 and the network for the stable operating condition. Note that the overlap already occurs at $\tau=1$. As $\tau$ decreases, the network DW shell shifts upward and gradually moves away from the converter DW shell, so the endpoint case $\tau=1$ is sufficient to reveal the violation of the DW shell separation condition.}
\label{fig:DW_stable}
\end{figure}

The proposed Stage-I optimization is then applied to the 11.2--14.4~Hz range left uncertified by the mixed gain-phase condition. Fig.~\ref{fig:QC_stable} shows the optimized common FQC margin $\epsilon^\star(\omega)$, obtained by solving~\eqref{Opt_stage1} in MATLAB with MOSEK. The margin is positive throughout the evaluated range, with a minimum of approximately 0.14 at 11.4~Hz. Hence, the full-multiplier condition in Proposition~\ref{pro:endpoint_full_multiplier} certifies this range, including 11.4--12.2~Hz, where DW shell separation conditions fail to certify stability. Together with the gain-phase certificate outside 11.2--14.4~Hz, this provides a stability certificate for the considered operating condition and demonstrates reduced conservatism in this case.

\begin{figure}[t]
\centering
\includegraphics[width=\columnwidth]{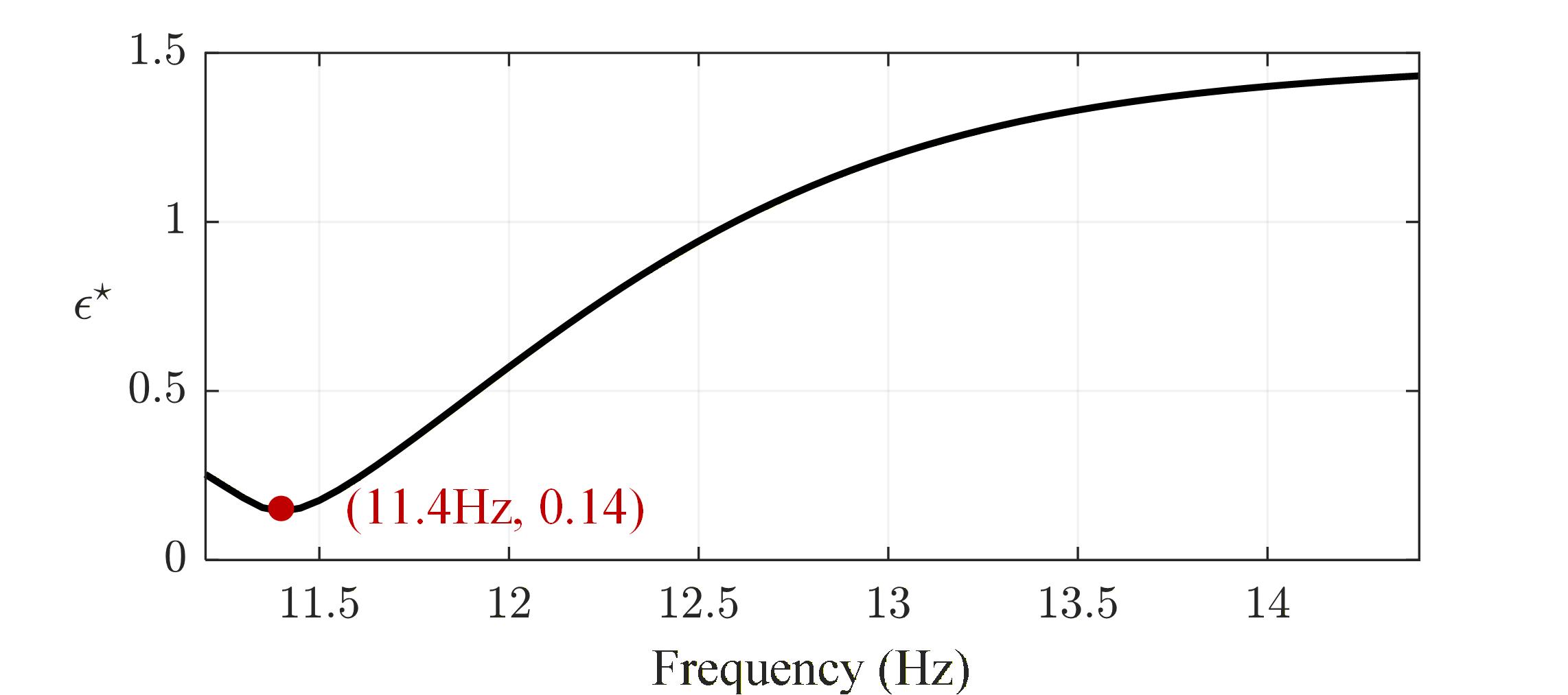}
\vspace{-6mm}
\caption{Common FQC margin of all GFL converters from 11.2Hz to 14.4Hz for the stable operating condition.}
\label{fig:QC_stable}
\end{figure}

To further validate the certificate, an electromagnetic transient (EMT) simulation is performed on this 68-bus system. At $t=0.2$~s, a $0.05$~p.u. voltage step is applied at the infinite bus and is subsequently restored at $t=0.22$~s. As shown in Fig.~\ref{fig:EMT_stable}, the converter terminal voltages and active powers converge back to the original equilibrium, consistent with the proposed stability certificate.

\begin{figure}[t]
\centering
\includegraphics[width=\columnwidth]{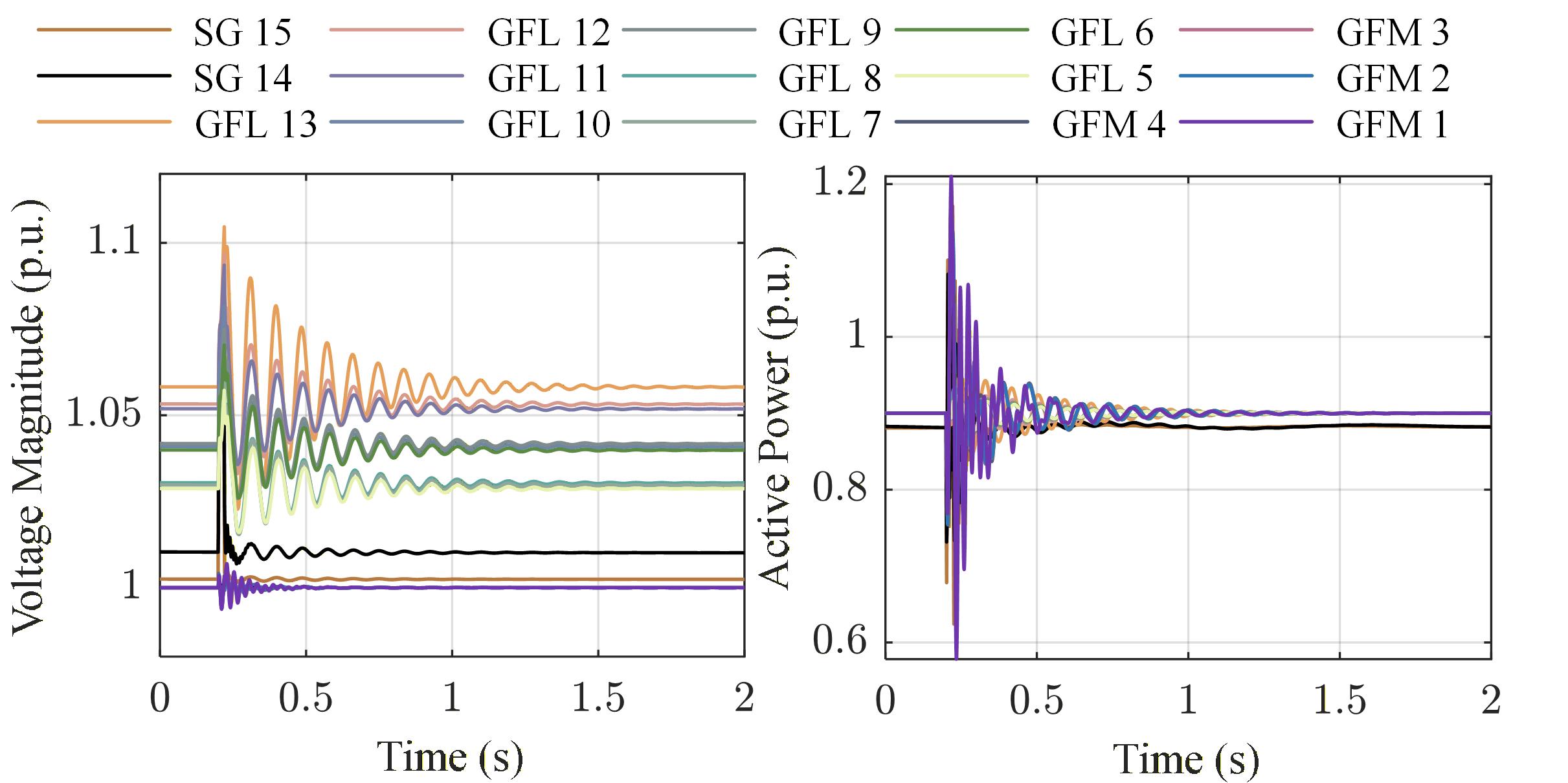}
\vspace{-8mm}
\caption{Time-domain results for the stable operating condition.}
\label{fig:EMT_stable}
\vspace{0mm}
\end{figure}

\subsection{Validation of Screening-Diagnosis Procedure}

We next consider an unstable operating condition to illustrate the complete screening-diagnosis procedure in Fig.~\ref{fig:flow_chart}. To emulate the situation of a weaker grid, all transmission line lengths are enlarged by a factor of 1.3. 
%This modification makes the interactions between the GFL converters and the network more pronounced.

The first step applies the mixed gain-phase condition as a low-cost preliminary screening tool. The gain and phase areas of the nine GFL converters are compared with those of the network, as shown in Fig.~\ref{fig:gain_phase_unstable}. Similar to the stable case in Section~V-A, the frequency axis can be divided into regions certified by the gain condition, regions certified by the phase condition, and an intermediate gray region where neither condition is satisfied. The gray region from 9.8~Hz to 14.4~Hz is therefore classified as uncertified after the first step, indicating that the conservative gain-phase test cannot determine the small-signal stability in certain frequency ranges.

\begin{figure}[t]
\centering
\includegraphics[width=\columnwidth]{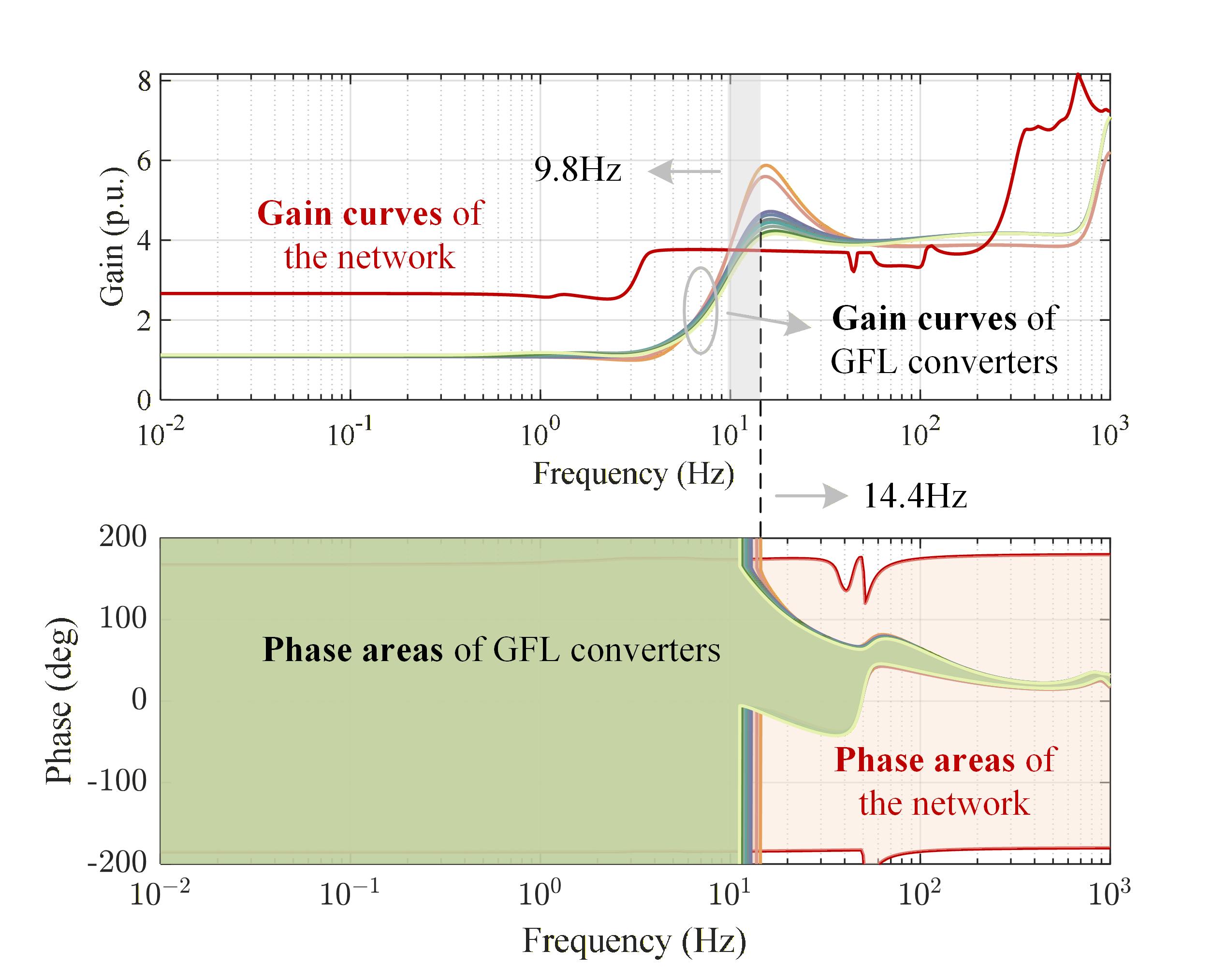}
\vspace{-8mm}
\caption{Gain curves and phase areas of the GFL converters and the network for the weak-grid operating condition.}
\label{fig:gain_phase_unstable}
\end{figure}

The second step applies Stage-I optimization to the uncertified frequency interval obtained from the gain-phase screening. By searching for the full multiplier, Stage-I further reduces the uncertified frequency range to [10.4~Hz, 11.6~Hz], as shown in the left panel of Fig.~\ref{fig:QC_unstable}. The remaining uncertified interval does not directly prove instability, but it identifies the frequency range in which the system may have a small-signal instability risk and therefore requires device-level diagnosis.

The third step applies Stage-II optimization over the remaining uncertified frequencies. As shown in the right panel of Fig.~\ref{fig:QC_unstable}, GFL converters 12 and 13 have negative indices and are identified as weak devices. The positive indices of GFL converters 5--11 indicate that their local FQC inequalities are satisfied under the same optimized multiplier.

\begin{figure}[t]
\centering
\includegraphics[width=\columnwidth]{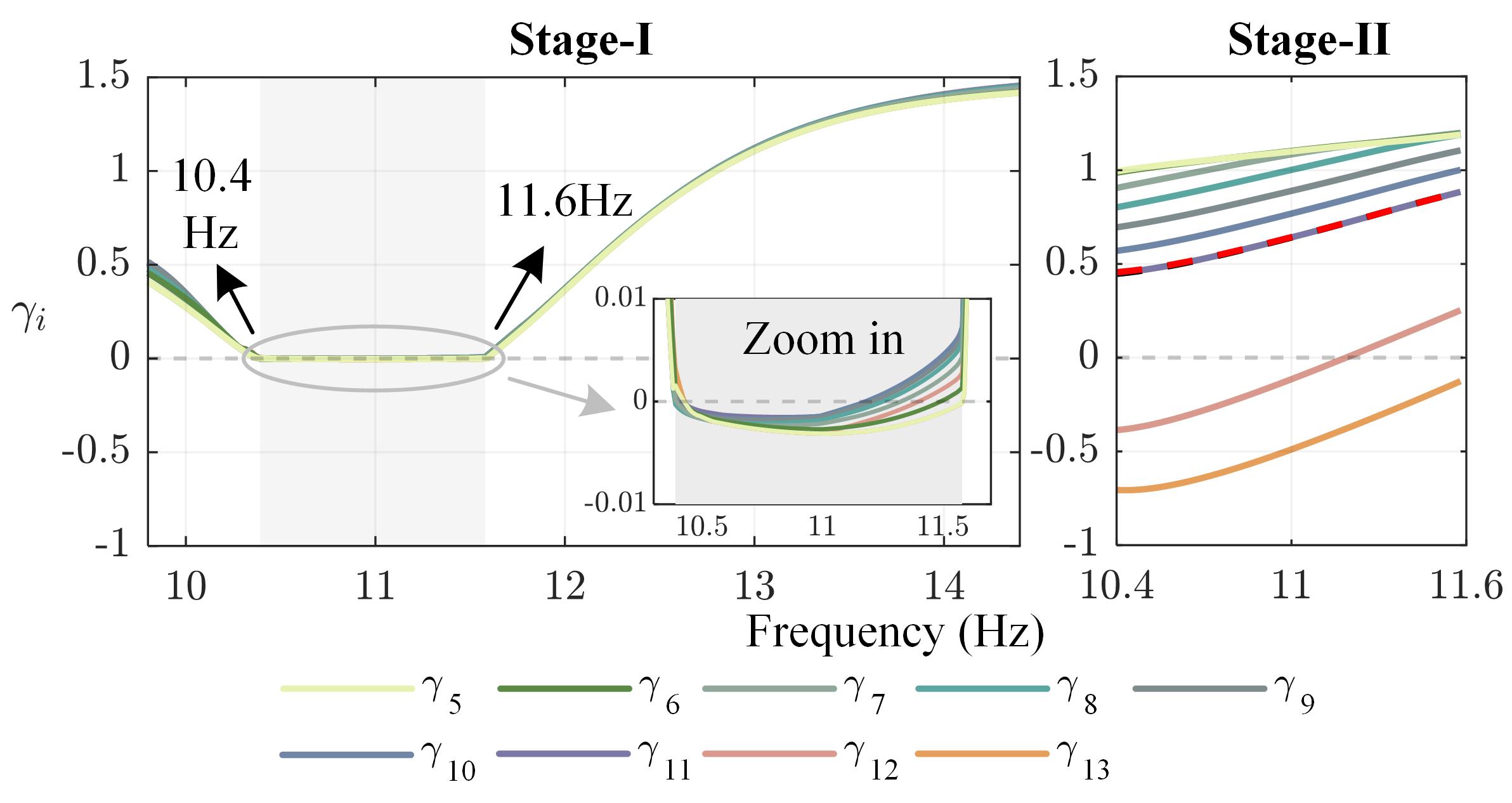}
\vspace{-8mm}
\caption{Results of the optimization for the weak-grid operating condition. Left: Stage-I optimization, with the gray-shaded region indicating the frequency range that remains uncertified. Right: Stage-II optimization applied to this uncertified frequency range for device-level diagnosis.}
\label{fig:QC_unstable}
\vspace{0mm}
\end{figure}

%This result is further verified by EMT simulation.
The single-circuit line connected to the terminal of GFL converter 13, with an impedance of $0.22$~(p.u.), corresponds to the weak-grid operating condition analyzed above. For EMT simulation, an additional parallel circuit, shown by the red dashed line in Fig.~\ref{fig:IEEE_68}, is added to obtain a stable pre-disturbance condition with an equivalent impedance of $0.11$~(p.u.). At $t=0.2$~s,  this additional circuit is tripped, increasing the equivalent impedance back to $0.22$~(p.u.) and restoring the weak-grid condition. As a result, the system indeed exhibits small-signal instability, as shown in the left panel of Fig.~\ref{fig:EMT_S_D_Procedure}. Guided by the Stage-II diagnosis, the dynamics of GFL converters 12 and 13 are reshaped to match the parameters of GFL converter 11, as indicated by the red dashed lines in Fig.~\ref{fig:QC_unstable}. After this reshaping, the proposed FQC-based stability condition can certify the stability of the modified system. The EMT simulation in the right panel of Fig.~\ref{fig:EMT_S_D_Procedure} under the same disturbance also shows that the system responses remain damped and converge to an equilibrium, confirming the effectiveness of the proposed screening-diagnosis procedure.

\begin{figure}[t]
\centering
\includegraphics[width=\columnwidth]{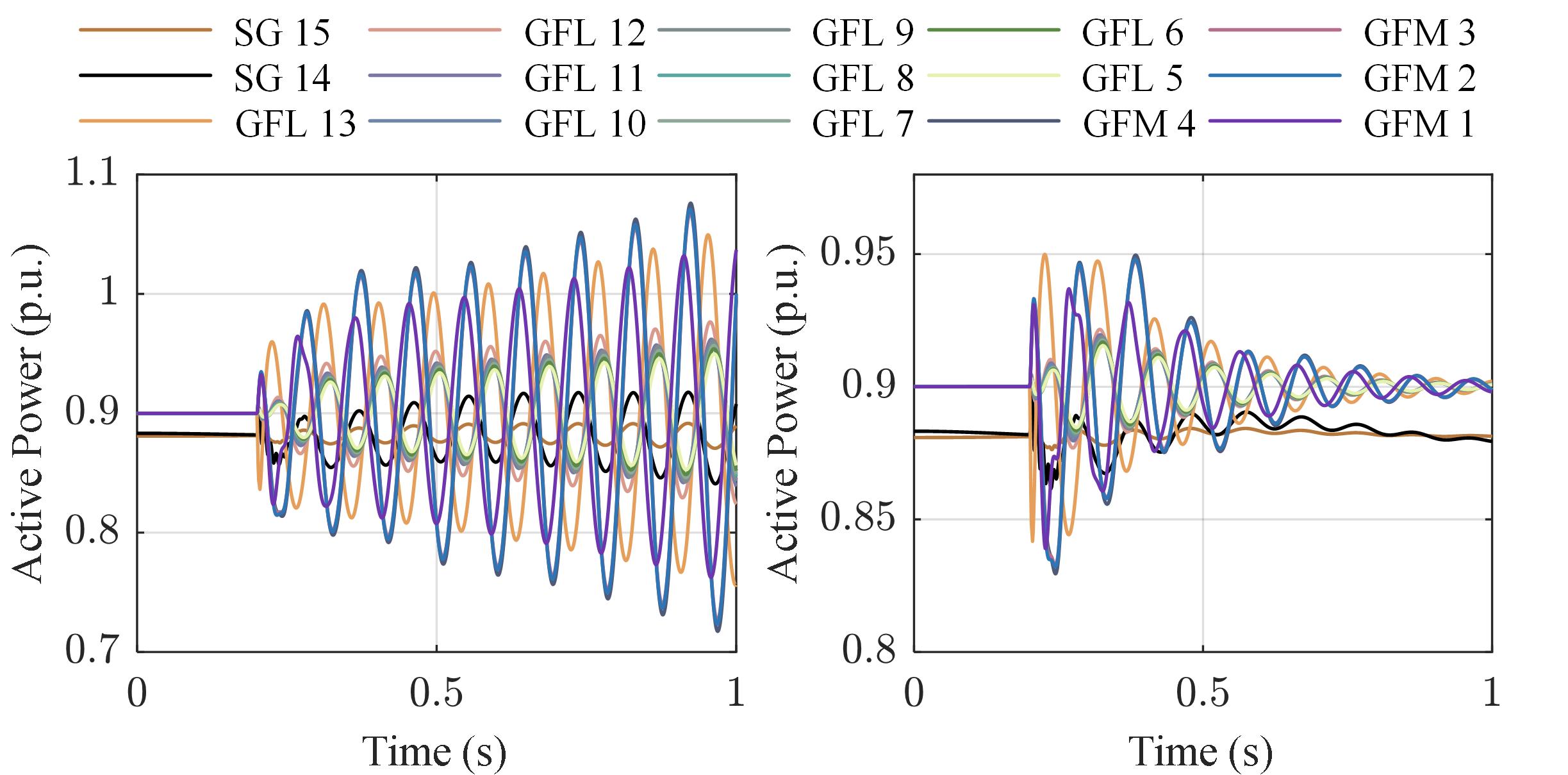}
\vspace{-6mm}
\caption{Time-domain results for the weak-grid condition. Left: active power response before reshaping. Right: active power response after reshaping.}
\label{fig:EMT_S_D_Procedure}
\vspace{3mm}
\end{figure}

\section{Conclusions}\label{sec:conclusion}
This paper proposed a scalable small-signal stability assessment method based on FQC for power electronics-dominated power systems. Several existing stability conditions have been shown to be special cases of the FQC-based condition with scalar multipliers, and the decentralization requirements of different stability conditions have been clarified. In particular, general geometric conditions require a common separating hyperplane for decentralized verification, whereas gain-phase conditions can be decentralized without such a construction. An exception arises for the \(x\)-\(z\) graph condition in networks with identical \(R/X\) ratios, where decentralized verification does not require a common-hyperplane search. To further reduce conservatism, scalar multipliers have been extended to full multipliers, leading to a full-multiplier scalable stability condition. Based on this condition, an optimization problem has been formulated for providing stability certification and performing device-level diagnosis. This optimization enables the computation of quantitative stability indices without any graphical inspection (required by DW shell stability conditions). It has been further combined with the mixed gain-phase condition to form a hierarchical screening-diagnosis procedure that improves the efficiency of stability assessment. Case studies have demonstrated that the proposed method is less conservative than existing criteria and can identify problematic devices associated with potential instability risks. Future work will extend the proposed FQC-based stability conditions toward systematic robust and optimal control design for heterogeneous devices, with the aim of reshaping device-level frequency-domain characteristics and enhancing the stability of power electronics-dominated power systems.

\newpage

\appendices
\section{Proof of Lemma~\ref{lem:QC_stability}}\label{app:proof_1}
\begin{proof}
For $\tau\in(0,1]$, suppose that $-1\in \lambda\bigl(\tau\bm{Y}_c(j\omega)\bm{Z}_g(j\omega)\bigr)$, where $\lambda(\cdot)$ denotes the set of eigenvalues of a matrix, which means that $\bm{I}+\tau\bm{Y}_c(j\omega)\bm{Z}_g(j\omega)$ is not invertible, i.e., $\det\bigl(\bm{I}+\tau\bm{Y}_c(j\omega)\bm{Z}_g(j\omega)\bigr)=0$. Then there exists a nonzero vector $\bm{w}(\omega,\tau)\in\mathbb{C}^{2n}$ such that $\bigl(\bm{I}+\tau\bm{Y}_c(j\omega)\bm{Z}_g(j\omega)\bigr)\bm{w}(\omega,\tau)=0$. Let $\bm{z}(\omega,\tau)=\tau\bm{Z}_g(j\omega)\bm{w}(\omega,\tau)$, then the previous equality becomes $\bm{w}(\omega,\tau)+\bm{Y}_c(j\omega)\bm{z}(\omega,\tau)=0$ and hence $\bm{w}(\omega,\tau)=-\bm{Y}_c(j\omega)\bm{z}(\omega,\tau)$. Define $\bm{\sigma}(\omega,\tau)=\bigl[\begin{smallmatrix}\bm{w}(\omega,\tau)\\\bm{z}(\omega,\tau)\end{smallmatrix}\bigr]$. After pre- and post-multiplying (\ref{IQC_1}) by $\bm{w}(\omega,\tau)^*$ and $\bm{w}(\omega,\tau)$, and (\ref{IQC_2}) by $\bm{z}(\omega,\tau)^*$ and $\bm{z}(\omega,\tau)$, respectively, we obtain $\bm{\sigma}(\omega,\tau)^*\bm{\Pi}(\omega,\tau)\bm{\sigma}(\omega,\tau)\geq0$ and $\bm{\sigma}(\omega,\tau)^* \bm{\Pi}(\omega,\tau)\bm{\sigma}(\omega,\tau)\leq -\epsilon(\omega,\tau)\bm{z}(\omega,\tau)^*\bm{z}(\omega,\tau)$. Since $\bm{z}(\omega,\tau)\neq 0$ and $\epsilon(\omega,\tau)>0$, this is a contradiction. Therefore, if conditions (\ref{IQC_1}) and (\ref{IQC_2}) are satisfied, then $\det\bigl(\bm{I}+\tau\bm{Y}_c(j\omega)\bm{Z}_g(j\omega)\bigr)\neq0$ for all $\omega\in[0,\infty)$ and $\tau\in(0,1]$.

Further, the condition $\det\bigl(\bm{I}+\tau\bm{Y}_c(j\omega)\bm{Z}_g(j\omega)\bigr)\neq0$ ensures that the characteristic loci do not pass through the origin during the continuous variation of $\tau$ from $0$ to $1$. This implies that the Nyquist plot $\det\bigl(\bm{I}+\bm{Y}_c(j\omega)\bm{Z}_g(j\omega)\bigr)$ does not encircle the origin as $\omega$ varies from $-\infty$ to $+\infty$. Therefore, according to the Nyquist stability criterion, the closed-loop interconnected system is stable when the open-loop systems $\bm{Y}_c(s)$ and $\bm{Z}_g(s)$ are stable. This completes the proof.
\end{proof}

\section{}\label{app:exist_stability_conditions}
This appendix focuses on how the stability conditions in Proposition~\ref{pro:exist_stability_conditions} can be derived from \eqref{IQC_DW_Zg} and \eqref{IQC_DW_Yc}.

(1) \emph{Davis-Wielandt (DW) shell separation}:

Let $a\in\mathbb{R}$, $b\in\mathbb{C}$, and $c\in\mathbb{R}$. By pre- and post-multiplying both sides by an arbitrary unit vector $\bm{u}\in\mathbb{C}^{2n}$, we obtain
\begin{equation*}
a+2\Re(b\bm{u}^*\tau\bm{Z}_g(j\omega)^*\bm{u})-c\|\tau\bm{Z}_g(j\omega)\bm{u}\|^2 \ge 0,
\end{equation*}
\begin{equation*}
a\|\bm{Y}_c(j\omega)\bm{u}\|^2-2\Re(b\bm{u}^*\bm{Y}_c(j\omega)\bm{u})-c\leq -\epsilon.
\end{equation*}

According to the definition of the DW shell in \eqref{Equ_DW}, let $z_g=\bm{u}^*\tau\bm{Z}_g(j\omega)\bm{u}$, $\nu_g=\|\tau\bm{Z}_g(j\omega)\bm{u}\|^2$, $z_c=\bm{u}^*\bm{Y}_c(j\omega)\bm{u}$, $\nu_c=\|\bm{Y}_c(j\omega)\bm{u}\|^2$. It then follows that $(z_g,\nu_g)\in\mathcal{DW}(\tau\bm{Z}_g(j\omega))$ and $(z_c,\nu_c)\in\mathcal{DW}(\bm{Y}_c(j\omega))$. By further invoking the properties of the DW shell in~\cite{ref_phantom_DW_shell}, namely, $(\hat{z}_g,\hat{\nu}_g)=(\frac{\overline{z_g}}{\nu_g},\frac{1}{\nu_g})\in\mathcal{DW}^{-1}(\tau\bm{Z}_g(j\omega))$ and $(\hat{z}_c,\hat{\nu}_c)=(-z_c,\nu_c)\in\mathcal{DW}(-\bm{Y}_c(j\omega))$, we have
\begin{subequations}
\begin{align}
a\hat{\nu}_g+2\Re(b\hat{z}_g) - c&\ge 0,
\label{IQC_DW_1}\\
a\hat{\nu}_c+2\Re(b\hat{z}_c) - c&\leq -\epsilon.
\label{IQC_DW_2}
\end{align}
\end{subequations}

By comparing \eqref{IQC_DW_1} and \eqref{IQC_DW_2}, define the affine function $F(z,\nu)=2\Re(bz)+a\nu$. It follows that $F(\hat{z}_g,\hat{\nu}_g)\ge c$ for all $(\hat{z}_g,\hat{\nu}_g)\in\mathcal{DW}^{-1}(\tau\bm{Z}_g(j\omega))$, whereas $F(\hat{z}_c,\hat{\nu}_c)\leq c-\epsilon$ for all $(\hat{z}_c,\hat{\nu}_c)\in\mathcal{DW}(-\bm{Y}_c(j\omega))$. Therefore, for any $x$ satisfying $c-\epsilon<x<c$, the hyperplane $H_d=\{(z,\nu)\in\mathbb{C}\times\mathbb{R}_+:2\Re(bz)+a\nu=x,a\in\mathbb{R},b\in\mathbb{C}\}$ strictly separates these two sets. Specifically, $\mathcal{DW}^{-1}(\tau\bm{Z}_g(j\omega))$ lies in the half-space $F(\hat{z}_g,\hat{\nu}_g) > x$, while $\mathcal{DW}(-\bm{Y}_c(j\omega))$ lies in the half-space $F(\hat{z}_c,\hat{\nu}_c) < x$. Hence, the two DW shells are disjoint, i.e., $\mathcal{DW}(-\bm{Y}_c(j\omega))\cap\mathcal{DW}^{-1}(\tau\bm{Z}_g(j\omega))=\emptyset$, recovering condition i) in Proposition~\ref{pro:exist_stability_conditions}.

(2) \emph{Numerical range and $x$-$z$ graph separation}:

Let $a=0$, $b\in\mathbb{C}$, $c\in\mathbb{R}$. Then, by the definition of the numerical range in \eqref{Equ_DW} and following a derivation similar to that for the DW-shell separation condition, we have
\begin{subequations}
\begin{align}
2\Re(b\hat{z}_g) - c&\ge 0, \quad \forall \hat{z}_g\in\mathcal{W}^{-1}(\tau\bm{Z}_g(j\omega)),
\label{IQC_W_1}\\
2\Re(b\hat{z}_c) - c&\leq -\epsilon, \quad \forall \hat{z}_c\in\mathcal{W}(-\bm{Y}_c(j\omega)).
\label{IQC_W_2}
\end{align}
\end{subequations}

Similarly, let $a, b, c\in\mathbb{R}$. Then, by the definition of the $x$-$z$ graph in \eqref{Equ_DW}, for all $(\Re(\hat{z}_g),\hat{\nu}_g)\in\mathcal{P}^{-1}(\tau\bm{Z}_g(j\omega))$ and $(\Re(\hat{z}_c),\hat{\nu}_c)\in\mathcal{P}(-\bm{Y}_c(j\omega))$, we have
\begin{subequations}
\begin{align}
a\hat{\nu}_g+2b\Re(\hat{z}_g) - c&\ge 0,
\label{IQC_P_1}\\
a\hat{\nu}_c+2b\Re(\hat{z}_c) - c&\leq -\epsilon.
\label{IQC_P_2}
\end{align}
\end{subequations}

Eqs.~\eqref{IQC_W_1} and~\eqref{IQC_W_2} imply that there exists a separating hyperplane $H_n=\{z\in\mathbb{C}:2\Re(bz)=x,b\in\mathbb{C}\}$ such that the two numerical range sets $\mathcal{W}^{-1}(\tau\bm{Z}_g(j\omega))$ and $\mathcal{W}(-\bm{Y}_c(j\omega))$ lie on opposite sides of $H_n$. Therefore, these two sets are disjoint, namely, $\mathcal{W}(-\bm{Y}_c(j\omega))\cap\mathcal{W}^{-1}(\tau\bm{Z}_g(j\omega))=\emptyset$, recovering condition ii) in Proposition~\ref{pro:exist_stability_conditions}. In the same way, according to \eqref{IQC_P_1} and \eqref{IQC_P_2}, there exists a separating hyperplane in the $x$-$z$ plane, i.e., $H_p=\{(\Re(z),\nu)\in\mathbb{R}\times \mathbb{R}_+:2b\Re(z)+a\nu=x,a,b\in\mathbb{R}\}$, such that the two $x$-$z$ graph sets $\mathcal{P}^{-1}(\tau\bm{Z}_g(j\omega))$ and $\mathcal{P}(-\bm{Y}_c(j\omega))$ lie on opposite sides of $H_p$. Therefore, these two sets are also disjoint, namely, $\mathcal{P}(-\bm{Y}_c(j\omega))\cap\mathcal{P}^{-1}(\tau\bm{Z}_g(j\omega))=\emptyset$, recovering condition iii) in Proposition~\ref{pro:exist_stability_conditions}.

(3) \emph{Scaled relative graph (SRG) separation}:

Let $a, b, c\in\mathbb{R}$ and $a=a_1-a_2$, $c=a_1+a_2$, $\xi=a_1+jb$. Combining this with the SRG transformation in \eqref{com_SRG}, we have
\begin{equation*}
-\Re(\overline{\xi}f(\tau\bm{Z}_g(j\omega))) - a_2\bm{I} \succeq 0,
\end{equation*}
\begin{equation*}
\Re(\xi f(\bm{Y}_c(j\omega))^*) - a_2\bm{I}\preceq -\epsilon\bm{I}.
\end{equation*}

Similar to the derivations of \eqref{IQC_DW_1}, \eqref{IQC_DW_2}, \eqref{IQC_W_1}, and \eqref{IQC_W_2}, let $z_g^{s}\in\mathcal{W}(f(\tau\bm{Z}_g(j\omega)))$ and $z_c^{s}\in\mathcal{W}(f(\bm{Y}_c(j\omega)))$. Together with $\hat{z}_g^{s}=-\overline{z_g^{s}}\in\mathcal{W}(f(\frac{1}{\tau}\bm{Y}_g(j\omega)))$ and $\hat{z}_c^{s}=\overline{z_c^{s}}\in\mathcal{W}(f(-\bm{Y}_c(j\omega)))$, we obtain
\begin{subequations}
\begin{align}
\Re(\xi\hat{z}_g^{s}) - a_2&\ge 0,
\label{IQC_SRG_1}\\
\Re(\xi\hat{z}_c^{s}) - a_2&\leq -\epsilon.
\label{IQC_SRG_2}
\end{align}
\end{subequations}

According to \eqref{IQC_SRG_1} and \eqref{IQC_SRG_2}, there exists a separating hyperplane $H_{s}=\{z\in\mathbb{C}:\Re(\xi z)=x,\xi\in\mathbb{C}\}$ such that the two sets $\mathcal{W}(f(\frac{1}{\tau}\bm{Y}_g(j\omega)))$ and $\mathcal{W}(f(-\bm{Y}_c(j\omega)))$ lie on opposite sides of $H_{s}$. This implies that $\mathcal{W}(f(\frac{1}{\tau}\bm{Y}_g(j\omega)))\cap\mathcal{W}(f(-\bm{Y}_c(j\omega)))=\emptyset$. Furthermore, according to \eqref{com_SRG}, for each point $z\in\mathcal{W}\!\bigl(f(\bm{M})\bigr)$, the corresponding point $g(z)\in\mathrm{SRG}(\bm{M})$ satisfies $\Re(z)=\frac{|g(z)|^2-1}{|g(z)|^2+1}$ and $\Im(z)=-\frac{2\Re(g(z))}{|g(z)|^2+1}$. Substituting these relationships into \eqref{IQC_SRG_1} and \eqref{IQC_SRG_2}, we obtain
\begin{subequations}
\begin{align}
a|r_g|^2-c+2b\Re(r_g)&\ge 0,
\label{IQC_SRG_3}\\
a|r_c|^2-c+2b\Re(r_c)&\le -\epsilon,
\label{IQC_SRG_4}
\end{align}
\end{subequations}
where $r_g=g(\hat{z}_g^{s})\in\mathrm{SRG}^{-1}(\tau\bm{Z}_g(j\omega))$ and $r_c=g(\hat{z}_c^{s})\in\mathrm{SRG}(-\bm{Y}_c(j\omega))$.

Similarly, \eqref{IQC_SRG_3} and \eqref{IQC_SRG_4} imply that the separating boundary $H_{r}=\{r\in\mathbb{C}:2b\Re(r)+a|r|^2=x,\ a,b\in\mathbb{R}\}$ separates the two sets $\mathrm{SRG}^{-1}(\tau\bm{Z}_g(j\omega))$ and $\mathrm{SRG}(-\bm{Y}_c(j\omega))$, which lie on opposite sides of $H_{r}$, i.e., $\mathrm{SRG}^{-1}(\tau\bm{Z}_g(j\omega))\cap\mathrm{SRG}(-\bm{Y}_c(j\omega))=\emptyset$, recovering condition iv) in Proposition~\ref{pro:exist_stability_conditions}.

(4) \emph{Small gain and small phase theorem}:

Let $a>0$, $b=0$, and $c\in\mathbb{R}$, we have
\begin{subequations}
\begin{align}
c&\leq\frac{a}{\sigma_1^2(\bm{Z}_g(j\omega))},
\label{IQC_gain_1}\\
\epsilon-c&\leq-a\sigma_1^2(\bm{Y}_c(j\omega)).
\label{IQC_gain_2}
\end{align}
\end{subequations}

Adding \eqref{IQC_gain_1} and \eqref{IQC_gain_2} yields $\sigma_1(\bm{Y}_c(j\omega))\sigma_1(\bm{Z}_g(j\omega))<1$, recovering condition v) in Proposition~\ref{pro:exist_stability_conditions}. Setting $b=0$ removes the dependence on the complex numerical-range coordinate and retains only the gain coordinate. The separating hyperplane therefore degenerates to the one-dimensional boundary $H_\sigma=\{v\in\mathbb{R}_+:av=x,a\in\mathbb{R}\}$, which reduces the separation test to a comparison of the converter-side and network-side gain bounds. 

Let $a=0$, $c=0$, and let $b\in\mathbb{C}$ be written as $b=b_1e^{jb_2}$, where $b_1>0$ and $b_2\in\mathbb{R}$. We have
\begin{subequations}
\begin{align}
\frac{\bm{Z}_g(j\omega)^*e^{jb_2}+e^{-jb_2}\bm{Z}_g(j\omega)}{2}&\succeq 0,\label{IQC_phase_1}\\
\frac{e^{jb_2}\bm{Y}_c(j\omega)+\bm{Y}_c(j\omega)^*e^{-jb_2}}{2}&\succ 0.\label{IQC_phase_2}
\end{align}
\end{subequations}

Eqs. \eqref{IQC_phase_1} and \eqref{IQC_phase_2} can be further developed as $\Re(e^{-jb_2}z_g^\phi) \ge 0$ and $\Re(e^{jb_2}z_c^\phi) > 0$, respectively, where $z_g^\phi\in\mathcal{W}(\bm{Z}_g(j\omega))$ and $z_c^\phi\in\mathcal{W}(\bm{Y}_c(j\omega))$. This implies that, at each frequency $\omega$, one seeks a rotation angle $b_2$ such that the numerical range of $\bm{Z}_g(j\omega)$ lies in the closed right half-plane after being rotated clockwise by $b_2$, while the numerical range of $\bm{Y}_c(j\omega)$ lies in the open right half-plane after being rotated counterclockwise by the same angle. Expanding these inequalities gives the corresponding boundaries \(H_\phi^\pm=\left\{z^\phi\in\mathbb C:\Re(z^\phi)\cos b_2\pm\Im(z^\phi)\sin b_2=0\right\}\), with \(H_\phi^+\) and \(H_\phi^-\) corresponding to the network and converter constraints, respectively. Both lines pass through the origin and are symmetric about the real axis.

Thus, $\bm{Z}_g(j\omega)$ and $\bm{Y}_c(j\omega)$ must be sectorial and satisfy $\phi_k(\bm{Z}_g(j\omega))-b_2\in[-\frac{\pi}{2},\frac{\pi}{2}]$ and $\phi_k(\bm{Y}_c(j\omega))+b_2\in(-\frac{\pi}{2},\frac{\pi}{2})$, respectively (\(k=1,\cdots,2n\)). Combining the upper and lower phase bounds separately yields $\phi_1(\bm{Y}_c(j\omega))+\phi_1(\bm{Z}_g(j\omega))<\pi$ and $\phi_{2n}(\bm{Y}_c(j\omega))+\phi_{2n}(\bm{Z}_g(j\omega))>-\pi$, recovering condition vi) in Proposition~\ref{pro:exist_stability_conditions}. It is worth noting that, if we further set $b_2=0$, this condition reduces to the passivity theorem.

For conditions i)–iv) in Proposition~\ref{pro:exist_stability_conditions}, the separation results in~\cite{ref_phantom_DW_shell} ensure the existence of the corresponding boundaries in Table~\ref{tab:comparison}. Choosing their coefficients and a positive separation margin gives feasible multipliers satisfying \eqref{IQC_DW_Zg} and \eqref{IQC_DW_Yc}. Conditions v) and vi) similarly admit the corresponding gain and phase multipliers. Closed-loop stability then follows directly from Corollary~\ref{coro:QC_stability}.

\section{Proof of Proposition~\ref{pro:scalable_xz_envelope}}\label{app:proof_2}
\begin{proof}
For each $\tau\in(0,1]$, choose $a=-1$, $b=S/(2\tau)$, and $c=P/\tau^2$. After substituting these multipliers into~\eqref{IQC_DW_Zg} and applying a congruence transformation with \(\bm Y_g=\bm Z_g^{-1}\), the network-side inequality is equivalent to \(-(\bm{Y}_g-\lambda_-\bm{I})(\bm{Y}_g-\lambda_+\bm{I})\succeq0\), which holds because every eigenvalue of $\bm{Y}_g$ lies in $[\lambda_-,\lambda_+]$. Under the same multipliers,~\eqref{Equ_local_xz_envelope} is precisely the local converter-side FQC. Since $\bm{Y}_c(j\omega)$ is block diagonal, these local inequalities assemble into~\eqref{IQC_DW_Yc}. Thus, the conditions of Corollary~\ref{coro:QC_stability} hold, which proves closed-loop stability.
\end{proof}

\section{Proof of Proposition~\ref{pro:decentralized_gain_phase}}\label{app:proof_5}
\begin{proof}
If~\eqref{Equ_decentralized_gain} holds, the block-diagonal structure of $\bm{Y}_c(j\omega)$ gives \(\sigma_1(\bm{Y}_c(j\omega))=\max_i\sigma_1(\bm{Y}_{c,i}(j\omega))\). Combining with \(\sigma_1(\bm{Z}_g(j\omega))=\sigma_{2n}^{-1}(\bm{Y}_g(j\omega))\) gives $\sigma_1(\bm{Y}_c(j\omega))\sigma_1(\bm{Z}_g(j\omega))<1$, satisfying the gain condition in Proposition~\ref{pro:exist_stability_conditions}.

If~\eqref{Equ_decentralized_phase} holds, conditions (a) and (b) place all converter phase areas within a common sector of width less than $\pi$, ensuring that the numerical range of the block-diagonal matrix $\bm{Y}_c(j\omega)$ lies in an open half-plane. Thus, $\bm{Y}_c(j\omega)$ is sectorial, and the block-diagonal structure gives $\phi_1(\bm{Y}_c(j\omega))=\max_i \phi_1(\bm{Y}_{c,i}(j\omega))$ and $\phi_{2n}(\bm{Y}_c(j\omega))=\min_i \phi_2(\bm{Y}_{c,i}(j\omega))$. Conditions (c) and (d) give $\phi_1(\bm{Y}_c(j\omega))+\phi_1(\bm{Z}_g(j\omega))<\pi$ and $\phi_{2n}(\bm{Y}_c(j\omega))+\phi_{2n}(\bm{Z}_g(j\omega))>-\pi$, satisfying the phase condition in Proposition~\ref{pro:exist_stability_conditions}. Thus, one of the sufficient conditions in Proposition~\ref{pro:exist_stability_conditions} holds at each frequency, proving closed-loop stability.
\end{proof}

\section{Proof of Proposition~\ref{pro:shared_full_multiplier}}\label{app:proof_3}
\begin{proof}
Since $\bm{Y}_g=\bm{Z}_g^{-1}$, a congruence transformation with $\bm{Y}_g$ makes the network-side inequality \eqref{IQC_3} equivalent to \eqref{Equ_Hg_tau}. Under the common full multiplier, the block-diagonal structure of $\bm{Y}_c$ makes the converter-side inequality \eqref{IQC_4} equivalent to \eqref{Equ_diag_Wi}. The latter holds if $\bm{W}_i\succeq\epsilon\bm{I}_2$ for every $i=1,\ldots,n$. Thus, the conditions of Corollary~\ref{coro:QC_stability} are satisfied, which proves closed-loop stability.
\end{proof} 

\section{Proof of Proposition~\ref{pro:endpoint_full_multiplier}}\label{app:proof_4}
\begin{proof}
For any $\tau\in(0,1]$, \(\bm{H}_g(\omega,\tau)\) in \eqref{Equ_Hg_tau} can be written as
\begin{equation*}
\bm{H}_g(\omega,\tau)
=\tau\bm{H}_g(\omega,1)
+(1-\tau)\bm{Y}_g^*\bm{A}\bm{Y}_g+(\tau-\tau^2)\bm{C}.
\end{equation*}

Since $\bm{H}_g(\omega,1)\succeq0$, $\bm{A}\succeq0$, $\bm{C}\succeq0$, and $\tau-\tau^2\ge0$ for $\tau\in(0,1]$, it follows that $\bm{H}_g(\omega,\tau)\succeq0$ for all $\tau\in(0,1]$. Since $\bm{Y}_c(j\omega)$ is block diagonal and the same matrices $\bm{A}_0$, $\bm{B}_0$, and $\bm{C}_0$ are used for all converters, the inequalities $\bm{W}_i\succeq\epsilon\bm{I}$ imply the converter-side inequality \eqref{IQC_4}. Therefore, the conditions of Corollary~\ref{coro:QC_stability} hold, which proves closed-loop stability.
\end{proof}
%\begin{IEEEbiographynophoto}{Jane Doe}
%Biography text here without a photo.
%\end{IEEEbiographynophoto}

%\begin{IEEEbiography}[{\includegraphics[width=1in,height=1.25in,clip,keepaspectratio]{fig1.png}}]{IEEE Publications Technology Team}
%In this paragraph you can place your educational, professional background and research and other interests.\end{IEEEbiography}

\end{document}